\documentclass[a4paper]{article}

\usepackage{pst-all}
\usepackage{fullpage}

\usepackage[utf8]{inputenc}
\usepackage{algpseudocode}

\usepackage{amsfonts}
\usepackage{amssymb}
\usepackage{amsthm}
\usepackage{amsmath}
\usepackage{nicefrac}

\newgray{gray}{0.6}
\newgray{lightgray}{0.8}

\newcommand{\N}{\mathbb{N}}
\newcommand{\R}{\mathbb{R}}
\newcommand{\C}{\mathcal{C}}
\newcommand{\cP}{\mathcal{P}}
\newcommand{\where}{\,\vline~}

\DeclareMathOperator{\vol}{vol}
\DeclareMathOperator{\ball}{B}
\DeclareMathOperator{\opt}{opt}
\DeclareMathOperator{\diam}{diam}
\DeclareMathOperator{\diamcost}{cost_{diam}}
\DeclareMathOperator{\rad}{rad}
\DeclareMathOperator{\radcost}{cost_{rad}}
\DeclareMathOperator{\drad}{drad}
\DeclareMathOperator{\dradcost}{cost_{drad}}
\DeclareMathOperator{\conf}{Conf}

\theoremstyle{plain}

\newtheorem{theorem}{Theorem}
\newtheorem{lemma}[theorem]{Lemma}
\newtheorem{corollary}[theorem]{Corollary}

\newtheorem{proposition}[theorem]{Proposition}

\newtheorem{definition}[theorem]{Definition}
\newtheorem{observation}[theorem]{Observation}
\newtheorem{problem}[theorem]{Problem}

\title{Analysis of Agglomerative Clustering\footnote{A preliminary version of this article appeared in \emph{Proceedings of the 28th International Symposium on Theoretical Aspects of Computer Science (STACS '11)}, March 2011, pp. 308--319. This article also appeared in \emph{Algorithmica}. The final publication is available at \texttt{http://link.springer.com/article/10.1007/s00453-012-9717-4}.}}
\author{
Marcel R. Ackermann\footnote{Schloss Dagstuhl -- Leibniz Center for Informatics, 66687 Wadern, Germany, \texttt{mra@dbis.uni-trier.de}, work done while at Department of Computer Science, University of Paderborn, Germany}
\and Johannes Bl\"omer\footnote{Department of Computer Science, University of Paderborn, 33098 Paderborn, Germany, \texttt{\{bloemer,kuntze\}@upb.de}}
\and Daniel Kuntze\footnotemark[3]
\and Christian Sohler\footnote{Department of Computer Science, TU Dortmund, 44221 Dortmund, Germany, \texttt{christian.sohler@tu-dortmund.de}\newline
For all four authors this research was supported by the German Research Foundation (DFG), grants BL 314/6-2 and SO 514/4-2.}
}
\date{}

\begin{document}
\maketitle
\thispagestyle{empty}

\begin{abstract}
The diameter $k$-clustering problem is the problem of partitioning a finite subset of $\R^d$ into $k$ subsets called clusters such that the maximum diameter of the clusters is minimized.
One early clustering algorithm that computes a hierarchy of approximate solutions to this problem (for all values of $k$) is the agglomerative clustering algorithm with the complete linkage strategy.
For decades, this algorithm has been widely used by practitioners.
However, it is not well studied theoretically.
In this paper, we analyze the agglomerative complete linkage clustering algorithm.
Assuming that the dimension $d$ is a constant, we show that for any $k$ the solution computed by this algorithm is an $O(\log k)$-approximation to the diameter $k$-clustering problem.
Our analysis does not only hold for the Euclidean distance but for any metric that is based on a norm.
Furthermore, we analyze the closely related $k$-center and discrete $k$-center problem.
For the corresponding agglomerative algorithms, we deduce an approximation factor of $O(\log k)$ as well.

\vspace{1ex}\noindent\emph{Keywords:} agglomerative clustering, hierarchical clustering, complete linkage, approximation guarantees
\end{abstract}

\section{Introduction}
Clustering is the process of partitioning a set of objects into subsets (called clusters) such that each subset contains similar objects and objects in different subsets are dissimilar.
There are many applications for clustering, including data compression \cite{pereira}, analysis of gene expression data \cite{eisen}, anomaly detection \cite{lee}, and structuring results of search engines \cite{broder}.
For every application, a proper objective function is used to measure the quality of a clustering.
One particular objective function is the largest diameter of the clusters.
If the desired number of clusters $k$ is given, we call the problem of minimizing this objective function the \emph{diameter $k$-clustering problem}.

One of the earliest and most widely used clustering strategies is agglomerative clustering.
The history of agglomerative clustering goes back at least to the 1950s (see for example \cite{florek,mcquitty}).
Later, biological taxonomy became one of the driving forces of cluster analysis.
In~\cite{sneath} the authors, who where the first biologists using computers to classify organisms, discuss several agglomerative clustering methods.

Agglomerative clustering is a bottom-up clustering process.
At the beginning, every input object forms its own cluster.
In each subsequent step, the two 'closest' clusters will be merged until only one cluster remains.
This clustering process creates a hierarchy of clusters, such that for any two different clusters $A$ and $B$ from possibly different levels of the hierarchy we either have $A \cap B = \emptyset$, $A\subset B$, or $B \subset A$.
Such a hierarchy is useful in many applications, for example, when one is interested in hereditary properties of the clusters (as in some bioinformatics applications) or if the exact number of clusters is a priori unknown.

In order to define the agglomerative strategy properly, we have to specify a distance measure between clusters.
Given a distance function between data objects, the following distance measures between clusters are frequently used.
In the \emph{single linkage strategy}, the distance between two clusters is defined as the distance between their closest pair of data objects.
Using this strategy is equivalent to computing a minimum spanning tree of the graph induced by the distance function using Kruskal's algorithm.
In case of the \emph{complete linkage strategy}, the distance between two clusters is defined as the distance between their farthest pair of data objects.
In the \emph{average linkage strategy} the distance is defined as the average distance between data objects from the two clusters.

\subsection{Related Work}
In this paper, we study the agglomerative clustering algorithm using the complete linkage strategy to find a hierarchical clustering of $n$ points from $\R^d$.
The running time is obviously polynomial in the description length of the input.
Therefore, our only goal in this paper is to give an approximation guarantee for the diameter $k$-clustering problem.
The approximation guarantee is given by a factor $\alpha$ such that the cost of the $k$-clustering computed by the algorithm is at most $\alpha$ times the cost of an optimal $k$-clustering.
Although the agglomerative complete linkage clustering algorithm is widely used, there are only few theoretical results considering the quality of the clustering computed by this algorithm.
It is known that there exists a certain metric distance function such that this algorithm computes a $k$-clustering with an approximation factor of $\Omega(\log k)$~\cite{dasgupta}.
However, prior to the analysis we present in this paper, no non-trivial upper bound for the approximation guarantee of the classical complete linkage agglomerative clustering algorithm was known, and deriving such a bound has been discussed as one of the open problems in~\cite{dasgupta}.

The diameter $k$-clustering problem is closely related to the \emph{$k$-center problem}.
In this problem, we are searching for $k$ centers and the objective is to minimize the maximum distance of any input point to the nearest center.
When the centers are restricted to come from the set of the input points, the problem is called the \emph{discrete $k$-center problem}.
It is known that for metric distance functions the costs of optimal solutions to all three problems are within a factor of $2$ from each other.

For the Euclidean case, we know that for fixed $k$, i.e. when we are not interested in a hierarchy of clusterings, the diameter $k$-clustering problem and the $k$-center problem are $\mathcal{NP}$-hard.
In fact, it is already $\mathcal{NP}$-hard to approximate both problems with an approximation factor below $1.96$ and $1.82$ respectively \cite{feder}.

Furthermore, there exist provably good approximation algorithms in this case.
For the discrete $k$-center problem, a simple $2$-approximation algorithm is known for metric spaces \cite{gonzalez}, which immediately yields a $4$-approximation algorithm for the diameter $k$-clustering problem.
For the $k$-center problem, a variety of results is known. For example, for the Euclidean metric in~\cite{BaHaIn02} a $(1+\epsilon)$-approximation algorithm with running time $2^{O(\nicefrac{k\log k}{\epsilon^2})}dn$ is shown.
This implies a $(2+\epsilon)$-approximation algorithm with the same running time for the diameter $k$-clustering problem. 

Also, for metric spaces a hierarchical clustering strategy with an approximation guarantee of $8$ for the discrete $k$-center problem is known~\cite{dasgupta}.
This implies an algorithm with an approximation guarantee of $16$ for the diameter $k$-clustering problem.

This paper as well as all of the above mentioned work is about static clustering, i.e. in the problem definition we are given the whole set of input points at once.
An alternative model of the input data is to consider sequences of points that are given one after another.
In \cite{charikar}, the authors discuss clustering in a so-called \emph{incremental clustering} model.
They give an algorithm with constant approximation factor that maintains a hierarchical clustering while new points are added to the input set.
Furthermore, they show a lower bound of $\Omega(\log k)$ for the agglomerative complete linkage algorithm and the diameter $k$-clustering problem.
However, since their model differs from ours, their results have no bearing on our results.

\subsection{Our contribution}
In this paper, we study the agglomerative complete linkage clustering algorithm and related algorithms for input sets $X\subset\R^d$.
To measure the distance between data points, we use a metric that is based on a norm, e.g., the Euclidean metric.
We prove that the agglomerative solution to the diameter $k$-clustering problem is an $O(\log k)$-approximation.
Here, the $O$-notation hides a constant that is doubly exponential in $d$.
This approximation guarantee holds for every level of the hierarchy computed by the algorithm.
That is, we compare each computed $k$-clustering with an optimal solution for that particular value of $k$.
These optimal $k$-clusterings do not necessarily form a hierarchy.
In fact, there are simple examples where optimal solutions have no hierarchical structure.

Our analysis also yields that if we allow $2k$ instead of $k$ clusters and compare the cost of the computed $2k$-clustering to an optimal solution with $k$ clusters, the approximation factor is independent of $k$ and depends only on $d$.
Moreover, the techniques of our analysis can be applied to prove stronger results for the $k$-center problem and the discrete $k$-center problem.
For the $k$-center problem, we derive an approximation guarantee that is logarithmic in $k$ and only singly exponential in $d$.
For the discrete $k$-center problem, we derive an approximation guarantee that is logarithmic in $k$ and the dependence on $d$ is only linear and additive.

Furthermore, we give almost matching upper and lower bounds for the one-dimensional case.
These bounds are independent of $k$.
For $d\geq2$ and the metric based on the $\ell_\infty$-norm, we provide a lower bound that exceeds the upper bound for $d=1$.
For $d\geq3$, we give a lower bound for the Euclidean case which is larger than the lower bound for $d=1$.
Finally, we construct instances providing lower bounds for any metric based on an $\ell_p$-norm with $1\leq p\leq\infty$.
However, the construction of these instances needs the dimension to depend on $k$.

\section{Preliminaries and problem definitions}
Throughout this paper, we consider input sets that are finite subsets of $\R^d$.
Our results hold for arbitrary metrics that are based on a norm, i.e., the distance \(\|x-y\|\) between two points \(x,y\in\R^d\) is measured using an arbitrary norm $\|\cdot\|$.
Readers who are not familiar with arbitrary metrics or are only interested in the Euclidean case, may assume that $\|\cdot\|_2$ is used, i.e. \(\|x-y\|=\sqrt{\sum_{i=1}^d(x_i-y_i)^2}\).
For $r\in\R$ and $y\in\R^d$, we denote the closed $d$-dimensional ball of radius $r$ centered at $y$ by \(\ball_r^d(y):=\left\{x\,|\,\|x-y\|\leq r\right\}\).

Given $k\in\N$ and a finite set \(X\subset\R^d\) with $k\leq|X|$, we say that \(\C_k=\{C_1,\ldots,C_k\}\) is a $k$-clustering of $X$ if the sets \(C_1,\ldots,C_k\) (called clusters) form a partition of $X$ into $k$ non-empty subsets.
We call a collection of $k$-clusterings of the same finite set $X$ but for different values of $k$ hierarchical, if it fulfills the following two properties.
First, for any \(1\leq k\leq|X|\) the collection contains at most one $k$-clustering.
Second, for any two of its clusterings $\C_i,\C_j$ with \(|\C_i|=i<j=|\C_j|\), every cluster in $\C_i$ is the union of one or more clusters from $\C_j$. 
A hierarchical collection of clusterings is called a hierarchical clustering.

We define the diameter of a finite and non-empty set \(C\subset\R^d\) to be \(\diam(C):=\max_{x,y\in C}\|x-y\|\).
Furthermore, we define the diameter cost of a $k$-clustering $\C_k$ as its largest diameter, i.e. \(\diamcost(\C_k):=\max_{C\in \C_k}\diam(C)\).
The radius of $C$ is defined as \(\rad(C):=\min_{y\in\R^d}\max_{x\in C}\|x-y\|\) and the radius cost of a $k$-clustering $\C_k$ is defined as its largest radius, i.e. \(\radcost(\C_k):=\max_{C\in \C_k}\rad(C)\).
Finally, we define the discrete radius of $C$ to be \(\drad(C):=\min_{y\in C}\max_{x\in C}\|x-y\|\) and the discrete radius cost of a $k$-clustering $\C_k$ is defined as its largest discrete radius, i.e. \(\dradcost(\C_k):=\max_{C\in \C_k}\drad(C)\).

\begin{problem}[discrete $k$-center]\label{drad_prob}
Given \(k\in\N\) and a finite set \(X\subset\R^d\) with $|X|\geq k$, find a $k$-clustering $\C_k$ of $X$ with minimal discrete radius cost.
\end{problem}

\begin{problem}[$k$-center]\label{rad_prob}
Given \(k\in\N\) and a finite set \(X\subset\R^d\) with $|X|\geq k$, find a $k$-clustering $\C_k$ of $X$ with minimal radius cost.
\end{problem}

\begin{problem}[diameter $k$-clustering]\label{diam_prob}
Given \(k\in\N\) and a finite set \(X\subset\R^d\) with $|X|\geq k$, find a $k$-clustering $\C_k$ of $X$ with minimal diameter cost.
\end{problem}

For our analysis of agglomerative clustering, we repeatedly use the volume argument stated in Lemma~\ref{volume_lem}.
This argument provides an upper bound on the minimum distance between two points from a finite set of points lying inside the union of finitely many balls.
For the application of this argument, the following definition is crucial.
\begin{definition}
Let $k\in\N$ and $r\in\R$. A set \(X\subset\R^d\) is called \emph{$(k,r)$-coverable} if there exist $y_1,\ldots,y_k\in\R^d$ with \(X\subseteq\bigcup_{i=1}^k\ball_r^d(y_i)\).
\end{definition}

\begin{lemma}\label{volume_lem}
Let $k\in\N$, $r\in\R$ and \(P\subset\R^d\) be finite and $(k,r)$-coverable with $|P|>k$.
Then, there exist distinct \(p,q\in P\) such that
\(\|p-q\|\leq4r\sqrt[d]{\frac{k}{|P|}}\).
\end{lemma}
\begin{proof}
Let $Z\subset\R^d$ with \(|Z|=k\) and \(P\subset\bigcup_{z\in Z}\ball_r^d(z)\).
We define $\delta$ to be the minimum distance between two points of $P$, i.e.
\(\delta:=\min_{\substack{p,q\in P\\p\neq q}}\|p-q\|\).
We assume for contradiction that \(u:=4r\sqrt[d]{\frac{k}{|P|}}<\delta\).
Since $|P|>k$ there exists $z\in Z$ with \(\left|\ball_r^d(z)\cap P\right|\geq2\).
It follows that $\delta\leq2r$ and hence, \(\frac{u}{2}<r\).
Note that for any $y\in\R^d$, $R\in\R$, and any norm $\|\cdot\|$, we have \(\vol\left(\ball_R^d(y)\right)=R^d\cdot V_d\), where $V_d$ is the volume of the $d$-dimensional unit ball $\ball_1^d(0)$ (see~\cite{webster}, Corollary 6.2.15).
Therefore, we deduce
\[\vol\left(\bigcup_{z\in Z}\ball_{r+\nicefrac{u}{2}}^d(z)\right)<\sum_{z\in Z}\vol\left(\ball_{2r}^d(z)\right)\leq k\cdot(2r)^d\cdot V_d.\]
Furthermore, since any $p\in P$ is contained in a ball $\ball_r^d(z)$ for some $z\in Z$, we conclude that any ball $\ball_{\nicefrac{u}{2}}^d(p)$ for $p\in P$ is contained in a ball $\ball_{r+\nicefrac{u}{2}}^d(z)$ for some $z\in Z$ (see Figure~\ref{pre_volem_fig}).
Thus, we deduce
\begin{equation}\label{volume_lem_eq}
\vol\left(\bigcup_{p\in P}\ball_{\nicefrac{u}{2}}^d(p)\right)<k\cdot(2r)^d\cdot V_d.
\end{equation}
However, since $u<\delta$, for any distinct \(p,q\in P\), we have \(\ball_{\nicefrac{u}{2}}^d(p)\cap\ball_{\nicefrac{u}{2}}^d(q)=\varnothing\).
Therefore, the total volume of the $|P|$ balls \(\ball_{\nicefrac{u}{2}}^d(p)\) is given by
\[\vol\left(\bigcup_{p\in P}\ball_{\nicefrac{u}{2}}^d(p)\right)=|P|\left(\frac{u}{2}\right)^dV_d=k\cdot(2r)^d\cdot V_d,\]
using the definition of $u$.
This contradicts (\ref{volume_lem_eq}).
We obtain $\delta\leq u$, which proves the lemma.
\end{proof}

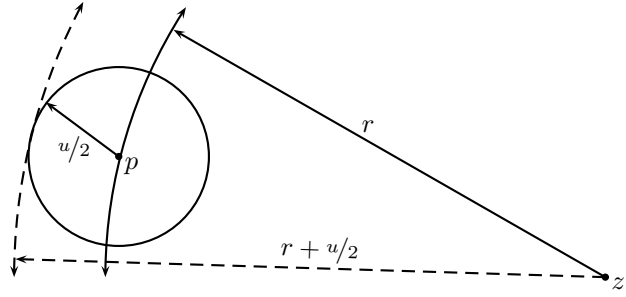
\begin{figure}
\begin{center}
	\psset{xunit=0.08cm,yunit=0.08cm,runit=0.08cm}
	\begin{pspicture}(0,0)(100,45)

 		\pscircle*(100,0){0.6}
 		\pscircle*(20,20){0.6}
 		\pscircle(20,20){15}
 		\put(101,-2){$z$}
 		\put(21,18){$p$}

		\psarc[linestyle=solid]{<->}(100,0){82.2}{147}{180}
		\psarc[linestyle=dashed]{<->}(100,0){97.2}{152}{180}

		\psline{->}(100,0)(29,41)
		\psline{->}(20,20)(8,29)
		\psline[linestyle=dashed]{->}(100,0)(3,3)

 		\rput(61,25){$r$}
 		\rput{-2}(53,4.5){$r+\nicefrac{u}{2}$}
 		\put(10,20){$\nicefrac{u}{2}$}

	\end{pspicture}
\end{center}
\caption{The volume argument.}
\label{pre_volem_fig}
\end{figure}

\section{Analysis}\label{analysis_sec}
In this section we analyze the agglomerative clustering algorithms for the (discrete) $k$-center problem and the diameter $k$-clustering problem.
As mentioned in the introduction, an agglomerative algorithm takes a bottom-up approach.
It starts with the $|X|$-clustering that contains one cluster for each input point and then successively merges two of the remaining clusters such that the cost of the resulting clustering is minimized.
That is, in each merge step the agglomerative algorithms for Problem~\ref{drad_prob}, Problem~\ref{rad_prob} and Problem~\ref{diam_prob} minimize the discrete radius, the radius and the diameter of the resulting cluster, respectively.

Our main objective is the agglomerative complete linkage clustering algorithm, which minimizes the diameter in every step.
Nevertheless, we start with the analysis of the agglomerative algorithm for the discrete $k$-center problem since it is the simplest one of the three.
Then we adapt our analysis to the $k$-center problem and finally to the diameter $k$-clustering problem.
In each case we need to introduce further techniques to deal with the increased complexity of the given problem.

We show that all three algorithms compute an $O(\log k)$-approximation for the particular corresponding clustering problem.
However, the dependency on the dimension which is hidden in the $O$-notation ranges from only linear and additive in case of the discrete $k$-center problem to a factor that is doubly exponential in case of the diameter $k$-clustering problem.

As mentioned in the introduction, the cost of optimal solutions to the three problems are within a factor of $2$ from each other.
That is, each algorithm computes an $O(\log k)$-approximation for all three problems.
However, we will analyze the proper agglomerative algorithm for each problem.

\subsection{Discrete $k$-center clustering}\label{drad_sec}
The agglomerative algorithm for the discrete $k$-center problem is stated as Algorithm~\ref{drad_algo}.
\begin{table}[t]
\centering
\begin{minipage}{.7\textwidth}
\footnotesize
\hrule
\textsc{AgglomerativeDiscreteRadius}$(X)$:\\
\begin{tabular}{rl}
$X$ & finite set of input points from $\R^d$\\
\end{tabular}
\hrule
\begin{algorithmic}[1]
\State $\C_{|X|}:=\left\{\,\{x\}\,|\,x\in X\right\}$
\For{$i=|X|-1,\ldots,1$}
\State find distinct clusters $A,B\in \C_{i+1}$ minimizing \(\drad(A\cup B)\)\label{drad_algo_min}
\State $\C_i:=(\C_{i+1}\setminus\{A,B\})\,\cup\,\{A\cup B\}$
\EndFor
\State\Return $\C_1,\ldots,\C_{|X|}$
\end{algorithmic}
\hrule
\end{minipage}
\vspace{1em}
\caption{The agglomerative algorithm for the discrete $k$-center problem.}
\label{drad_algo}
\end{table}
Given a finite set $X\subset\R^d$ of input points, the algorithm computes hierarchical $k$-clusterings for all values of $k$ between $1$ and $|X|$.
We denote them by $\C_1,\ldots,\C_{|X|}$.
Throughout this section, \emph{cost} always means discrete radius cost.
$\opt_k$ refers to the cost of an optimal discrete $k$-center clustering of $X\subset\R^d$ where $k\in\N$ with \(k\leq|X|\), i.e. the cost of an optimal solution to Problem~\ref{drad_prob}.

The following theorem states our result for the discrete $k$-center problem.
\begin{theorem}\label{drad_result}
Let \(X\subset\R^d\) be a finite set of points.
Then, for all $k\in\N$ with $k\leq|X|$, the partition $\C_k$ of $X$ into $k$ clusters as computed by Algorithm~\ref{drad_algo} satisfies
\[\dradcost(\C_k)<(20d+2\log_2(k)+2)\cdot\opt_k,\]
where $\opt_k$ denotes the cost of an optimal solution to Problem~\ref{drad_prob}.
\end{theorem}

Since any cluster $C$ is contained in a ball of radius $\drad(C)$, we have that $X$ is $(k,\dradcost(\C_k))$-coverable for any $k$-clustering $\C_k$ of $X$.
It follows, that $X$ is $(k,\opt_k)$-coverable.
This fact, as well as the following observation about the greedy strategy of Algorithm~\ref{drad_algo}, will be used frequently in our analysis..

\begin{observation}\label{drad_greedy_obs}
The cost of all computed clusterings is equal to the discrete radius of the cluster created last.
Furthermore, the discrete radius of the union of any two clusters is always an upper bound for the cost of the clustering to be computed next.
\end{observation}

We prove Theorem~\ref{drad_result} in two steps.
First, Proposition~\ref{drad_2k_prop} in Section~\ref{drad_2k_sec} provides an upper bound to the cost of the intermediate $2k$-clustering.
This upper bound is independent of $k$ and $|X|$, only linear in $d$ and may be of independent interest.
In its proof, we use Lemma~\ref{volume_lem} to bound the distance between the centers of pairs of remaining clusters.
The cost of merging such a pair gives an upper bound to the cost of the next merge step.
Therefore, we can bound the discrete radius of the created cluster by the sum of the larger of the two clusters' discrete radii and the distance between their centers.

Second, in Section~\ref{drad_rem_sec}, we analyze the remaining $k$ merge steps of Algorithm~\ref{drad_algo} down to the computation of the $k$-clustering.
There, we no longer need to apply the volume argument from Lemma~\ref{volume_lem} to bound the distance between two cluster centers.
It will be replaced by a very simple bound that is already sufficient.
Analogously to the first step, this leads to a bound for the cost of merging a pair of clusters.

\subsubsection{Analysis of the $2k$-clustering}\label{drad_2k_sec}
\begin{proposition}\label{drad_2k_prop}
Let \(X\subset\R^d\) be finite.
Then, for all $k\in\N$ with $2k\leq|X|$, the partition $\C_{2k}$ of $X$ into $2k$ clusters as computed by Algorithm~\ref{drad_algo} satisfies
\[\dradcost(\C_{2k})<20d\cdot\opt_k,\]
where $\opt_k$ denotes the cost of an optimal solution to Problem~\ref{drad_prob}.
\end{proposition}

To prove Proposition~\ref{drad_2k_prop}, we divide the merge steps of Algorithm~\ref{drad_algo} into phases, each reducing the number of remaining clusters by one fourth.
The following lemma bounds the increase of the cost during a single phase by an additive term.

\begin{lemma}\label{drad_2k_phaselem}
Let $m\in\N$ with \(2k<m\leq|X|\).
Then,
\begin{equation*}\label{drad_2k_phaseeq}
\dradcost\left(\C_{\left\lfloor\frac{3m}{4}\right\rfloor}\right)<\dradcost(\C_m)+4\sqrt[d]{\frac{2k}{m}}\cdot\opt_k.
\end{equation*}
\end{lemma}
\begin{proof}
Let $R:=\dradcost(\C_m)$.
From every cluster $C\in\C_m$, we fix a center $p_C\in C$ with $C\subset\ball_R^d(p_C)$.

Let $t:=\left\lfloor\frac{3m}{4}\right\rfloor$.
Then, $\C_m\cap\C_{t+1}$ is the set of clusters from $\C_m$ that still exist $\left\lceil\frac{m}{4}\right\rceil-1$ merge steps after the computation of $\C_m$.
In each iteration of its loop, the algorithm can merge at most two clusters from $\C_m$.
Thus, \(|\C_m\cap\C_{t+1}|>\frac{m}{2}\).

Let \(P:=\{p_C\where C\in C_m\cap \C_{t+1}\}\).
Then, $|P|=|\C_m\cap\C_{t+1}|>\frac{m}{2}>k$.
Since $X$ is $(k,\opt_k)$-coverable, so is \(P\subset X\).
Therefore, by Lemma~\ref{volume_lem}, there exist distinct $C_1,C_2\in\C_m\cap\C_{t+1}$ such that
\(\|p_{C_1}-p_{C_2}\|\leq4\sqrt[d]{\frac{2k}{m}}\cdot\opt_k\).
Then, the distance from $p_{C_1}$ to any $q\in C_2$ is at most $4\sqrt[d]{\frac{2k}{m}}\cdot\opt_k+R$.
We conclude that merging $C_1$ and $C_2$ would result in a cluster whose discrete radius can be upper bounded by
\(\drad(C_1\cup C_2)<\dradcost(\C_m)+4\sqrt[d]{\frac{2k}{m}}\cdot\opt_k\) (see Figure~\ref{drad_merge_fig}).
The result follows using $C_1,C_2\in\C_{t+1}$ and Observation~\ref{drad_greedy_obs}.
\end{proof}

\begin{figure}
\begin{center}
	\psset{xunit=0.08cm,yunit=0.08cm,runit=0.08cm}
	\begin{pspicture}(0,0)(100,45)

 		\pscircle*(20,20){0.6}
 		\pscircle*(80,25){0.6}
 		\pscircle(20,20){20}
 		\pscircle(80,25){20}
 		\put(15,17){$p_{C_1}$}
 		\put(81,23){$p_{C_2}$}

		\psarc[linestyle=dashed]{<->}(20,20){80}{-15}{20}
		\psline[linestyle=dashed]{->}(20,20)(99,10)

		\psline{->}(20,20)(2,29)
		\psline{->}(80,25)(98,34)
		\psline{|-|}(20,20)(80,25)

 		\put(8,27){$\dradcost(\C_m)$}
 		\put(69,33){$\dradcost(\C_m)$}

	\end{pspicture}
\end{center}
\caption{$\drad(C_1\cup C_2)<\dradcost(\C_m)+\|p_{C_1}-p_{C_2}\|$.}
\label{drad_merge_fig}
\end{figure}

To prove Proposition~\ref{drad_2k_prop}, we apply Lemma~\ref{drad_2k_phaselem} for $\left\lceil\log_\frac{4}{3}\frac{|X|}{2k}\right\rceil$ consecutive phases.
\begin{proof}[Proof of Proposition~\ref{drad_2k_prop}]
Let \(u:=\left\lceil\log_\frac{4}{3}\frac{|X|}{2k}\right\rceil\) and define \(m_i:=\left\lceil\left(\frac{3}{4}\right)^i|X|\right\rceil\) for all $i=0,\ldots,u$.
Then, \(m_u\leq2k\) and \(m_i>2k\) for all $i=0,\ldots,u-1$.
Since \(\left\lfloor\frac{3m_i}{4}\right\rfloor=\left\lfloor\frac{3}{4}\left\lceil\left(\frac{3}{4}\right)^i|X|\right\rceil\right\rfloor\leq\left\lfloor\left(\frac{3}{4}\right)^{i+1}|X|+\frac{3}{4}\right\rfloor\leq\left\lceil\left(\frac{3}{4}\right)^{i+1}|X|\right\rceil=m_{i+1}\) and Algorithm~\ref{drad_algo} uses a greedy strategy, we get \(\dradcost(\C_{m_{i+1}})\leq\dradcost(\C_{\left\lfloor\frac{3m_i}{4}\right\rfloor})\) for all $i=0,\ldots,u-1$.
Combining this with Lemma~\ref{drad_2k_phaselem} (applied to $m=m_i$), we obtain
\[\dradcost(\C_{m_{i+1}})<\dradcost(\C_{m_i})+4\sqrt[d]{\frac{2k}{m_i}}\cdot\opt_k.\]
By repeatedly applying this inequality for $i=0,\ldots,u-1$ and using \(\dradcost(\C_{2k})\leq\dradcost(\C_{m_u})\) and \(\dradcost(\C_{m_0})=0\), we deduce
\[\dradcost\left(\C_{2k}\right)<\sum_{i=0}^{u-1}\left(4\sqrt[d]{\frac{2k}{m_i}}\cdot\opt_k\right)<4\sqrt[d]{\frac{2k}{|X|}}\cdot\sum_{i=0}^{u-1}\sqrt[d]{\left(\frac{4}{3}\right)^i}\cdot\opt_k.\]
Solving the geometric series and using $u-1<\log_\frac{4}{3}\frac{|X|}{2k}$ leads to
\begin{equation}\label{drad_gseq}
\dradcost\left(\C_{2k}\right)<4\sqrt[d]{\frac{2k}{|X|}}\cdot\frac{\sqrt[d]{\left(\frac{4}{3}\right)^u}-1}{\sqrt[d]{\frac{4}{3}}-1}\cdot\opt_k<\frac{4\sqrt[d]{\frac{4}{3}}}{\sqrt[d]{\frac{4}{3}}-1}\cdot\opt_k.
\end{equation}
By taking only the first two terms of the series expansion of the exponential function, we get
\(\sqrt[d]{\frac{4}{3}}=e^{\frac{\ln\frac{4}{3}}{d}}>1+\frac{\ln\frac{4}{3}}{d}\).
Substituting this bound into (\ref{drad_gseq}) gives
\[\dradcost\left(\C_{2k}\right)<\frac{4\sqrt[d]{\frac{4}{3}}}{\ln\frac{4}{3}}d\cdot\opt_k<20d\cdot\opt_k.\]
\end{proof}

\subsubsection{Analysis of the remaining merge steps}\label{drad_rem_sec}
The analysis of the remaining merge steps introduces the $O(\log k)$ term to the approximation factor of our result.
It is similar to the analysis used in the proof of Proposition~\ref{drad_2k_prop}.
Again, we divide the merge steps into phases.
However, this time one phase consists of one half of the remaining merge steps.
Furthermore, we are able to replace the volume argument from Lemma~\ref{volume_lem} by a simpler bound.
More precisely, as long as there are more than $k$ clusters left, we are able to find a pair of clusters whose centers lie in the same cluster of an optimal $k$-clustering.
That is, the distance between the centers is at most two times the discrete radius of the common cluster in the optimal clustering.
The following lemma bounds the increase of the cost during a single phase.

\begin{lemma}\label{drad_rem_phaselem}
Let $m\in\N$ with \(k<m\leq|X|\).
Then,
\[\dradcost(\C_{k+\left\lfloor\frac{m-k}{2}\right\rfloor})<\dradcost(\C_m)+2\opt_k.\]
\end{lemma}
\begin{proof}
Let $R:=\dradcost(\C_m)$.
From every cluster $C\in\C_m$, we fix a center $p_C\in C$ with $C\subset\ball_R^d(p_C)$.

Let $t:=k+\left\lfloor\frac{m-k}{2}\right\rfloor$.
Then, $\C_m\cap \C_{t+1}$ is the set of clusters from $\C_m$ that still exist $\left\lceil\frac{m-k}{2}\right\rceil-1$ merge steps after the computation of $\C_m$.
In each iteration of its loop, the algorithm can merge at most two clusters from $\C_m$.
Thus, \(|\C_m\cap \C_{t+1}|>k\).

Let \(P:=\{p_C\where C\in C_m\cap \C_{t+1}\}\).
Since $X$ is $(k,\opt_k)$-coverable, so is \(P\subset X\).
Therefore, using $|P|>k$ it follows that there exist distinct $C_1,C_2\in\C_m\cap\C_{t+1}$ such that $p_{C_1}$ and $p_{C_2}$ are contained in the same ball of radius $\opt_k$, i.e.
\(\|p_{C_1}-p_{C_2}\|\leq2\opt_k\).
Then, the distance from $p_{C_1}$ to any $q\in C_2$ is at most $2\opt_k+R$.
We conclude that merging $C_1$ and $C_2$ would result in a cluster whose discrete radius can be upper bounded by
\(\drad(C_1\cup C_2)<\dradcost(\C_m)+2\opt_k\)  (see Figure~\ref{drad_merge_fig}).
The result follows using $C_1,C_2\in\C_{t+1}$ and Observation~\ref{drad_greedy_obs}.
\end{proof}

To prove Theorem~\ref{drad_result}, we apply Lemma~\ref{drad_rem_phaselem} for about $\log k$ consecutive phases.
\begin{proof}[Proof of Theorem~\ref{drad_result}]
Let  $\varepsilon>0$ and \(u:=\left\lceil\log_2(k)+\varepsilon\right\rceil\) such that \(\log_2k<u\leq\log_2(k)+1\).
Furthermore, define \(m_i:=k+\left\lfloor\left(\frac{1}{2}\right)^ik\right\rfloor\) for all $i=0,\ldots,u$.
Then, \(m_u=k\) and \(m_i>k\) for all $i=0,\ldots,u-1$.
Since \(k+\left\lfloor\frac{m_i-k}{2}\right\rfloor=k+\left\lfloor\frac{1}{2}\left\lfloor\left(\frac{1}{2}\right)^ik\right\rfloor\right\rfloor\leq k+\left\lfloor\left(\frac{1}{2}\right)^{i+1}k\right\rfloor=m_{i+1}\) and Algorithm~\ref{drad_algo} uses a greedy strategy, we deduce \(\dradcost(\C_{m_{i+1}})\leq\dradcost\left(\C_{k+\left\lfloor\frac{m_i-k}{2}\right\rfloor}\right)\) for all $i=0,\ldots,u-1$.
Combining this with Lemma~\ref{drad_rem_phaselem} (applied to $m=m_i$), we obtain
\[\diamcost(\C_{m_{i+1}})<\dradcost(\C_{m_i})+2\opt_k.\]
By repeatedly applying this inequality for $i=0,\ldots,u-1$ and using \(m_0=2k\), we get
\[\dradcost(\C_k)<\dradcost(\C_{2k})+\sum_{i=0}^{u-1}2\opt_k\leq\dradcost(\C_{2k})+(2\log_2(k)+2)\cdot\opt_k.\]
Hence, the result follows using Proposition~\ref{drad_2k_prop}.
\end{proof}

\subsection{$k$-center clustering}\label{rad_sec}
The agglomerative algorithm for the $k$-center problem is stated as Algorithm~\ref{rad_algo}.
\begin{table}[t]
\centering
\begin{minipage}{.7\textwidth}
\footnotesize
\hrule
\textsc{AgglomerativeRadius}$(X)$:\\
\begin{tabular}{rl}
$X$ & finite set of input points from $\R^d$\\
\end{tabular}
\hrule
\begin{algorithmic}[1]
\State $\C_{|X|}:=\left\{\,\{x\}\,|\,x\in X\right\}$
\For{$i=|X|-1,\ldots,1$}
\State find distinct clusters $A,B\in \C_{i+1}$ minimizing \(\rad(A\cup B)\)\label{rad_algo_min}
\State $\C_i:=(\C_{i+1}\setminus\{A,B\})\,\cup\,\{A\cup B\}$
\EndFor
\State\Return $\C_1,\ldots,\C_{|X|}$
\end{algorithmic}
\hrule
\end{minipage}
\vspace{1em}
\caption{The agglomerative algorithm for the $k$-center problem.}
\label{rad_algo}
\end{table}
The only difference to Algorithm~\ref{drad_algo} is the minimization of the radius instead of the discrete radius in Step~\ref{rad_algo_min}.

In the following, \emph{cost} always means radius cost and $\opt_k$ refers to the cost of an optimal $k$-center clustering of $X\subset\R^d$ where $k\in\N$ with \(k\leq|X|\).

\begin{observation}[analogous to Observation~\ref{drad_greedy_obs}]\label{rad_greedy_obs}
The cost of all computed clusterings is equal to the radius of the cluster created last.
Furthermore, the radius of the union of any two clusters is always an upper bound for the cost of the clustering to be computed next.
\end{observation}

The following theorem states our result for the $k$-center problem.
\begin{theorem}\label{rad_result}
Let \(X\subset\R^d\) be a finite set of points.
Then, for all $k\in\N$ with $k\leq|X|$, the partition $\C_k$ of $X$ into $k$ clusters as computed by Algorithm~\ref{rad_algo} satisfies
\[\radcost(\C_k)=O(\log k)\cdot\opt_k,\]
where $\opt_k$ denotes the cost of an optimal solution to Problem~\ref{rad_prob}, and the constant hidden in the $O$-notation is singly exponential in the dimension $d$.
\end{theorem}

Theorem~\ref{rad_result} holds for any particular tie-breaking strategy.
However, to keep the analysis simple, we assume that there are no ties.
That is, we assume that for any input set $X$ the clusterings computed by Algorithm~\ref{rad_algo} are uniquely determined.
As in the proof of Theorem~\ref{drad_result}, we first show a bound for the cost of the intermediate $2k$-clustering.
However, we have to apply a different analysis.
As a consequence, the dependency on the dimension increases from linear and additive to a singly exponential factor.

\subsubsection{Analysis of the $2k$-clustering}\label{rad_2k_sec}
\begin{proposition}\label{rad_2k_prop}
Let \(X\subset\R^d\) be finite.
Then, for all $k\in\N$ with $2k\leq|X|$, the partition $\C_{2k}$ of $X$ into $2k$ clusters as computed by Algorithm~\ref{rad_algo} satisfies
\[\radcost(\C_{2k})<24d\cdot e^{24d}\cdot\opt_k,\]
where $\opt_k$ denotes the cost of an optimal solution to Problem~\ref{rad_prob}.
\end{proposition}

Just as in the analysis of Algorithm~\ref{drad_algo}, we divide the merge steps of Algorithm~\ref{rad_algo} into phases, such that in each phase the number of remaining clusters is reduced by one fourth.
Like in the discrete case, the input points are $(k,\opt_k)$-coverable.
However, centers corresponding to an intermediate solution computed by Algorithm~\ref{rad_algo} need not be covered by the $k$ balls induced by an optimal solution.
As a consequence, we are no longer able to apply Lemma~\ref{volume_lem} on the centers as in the discrete case.

To bound the increase of the cost during a single phase, we cover the remaining clusters at the beginning of a phase by a set of overlapping balls.
Each of the clusters is completely contained in one of these balls that all have the same radius.
Furthermore, the number of remaining clusters will be at least twice the number of these balls.
It follows that there are many pairs of clusters that are contained in the same ball.
Then, as long as the existence of at least one such pair can be guaranteed, the radius of the cluster created next can be bounded by the radius of the balls.
The following lemma will be used to bound the increase of the cost during a single phase.

\begin{lemma}\label{rad_2k_phaselem}
Let $m\in\N$ with \(2k<m\leq|X|\).
Then,
\begin{equation*}\label{rad_2k_phaseeq}
\radcost\left(\C_{\left\lfloor\frac{3m}{4}\right\rfloor}\right)<\left(1+6\sqrt[d]{\frac{2k}{m}}\right)\cdot\radcost(\C_m)+6\sqrt[d]{\frac{2k}{m}}\cdot\opt_k.
\end{equation*}
\end{lemma}

\begin{proof}
Let \(\cP=\{P_1,\ldots,P_k\}\) be an optimal $k$-clustering of $X$.
We fix $y_1,\ldots,y_k\in\R^d$ such that $P_i\subset\ball_{\opt_k}(y_i)$ for $i=1,\ldots,k$.
For any $C\in\C_m$ let $z_C\in\R^d$ such that $C\subset\ball_{\radcost(\C_m)}(z_C)$.
It follows that each $z_C$ is contained in at least one of the balls $\ball_{\opt_k+\radcost(\C_m)}(y_i)$ for $i=1,\ldots,k$ (see Figure~\ref{rad_2k_phasefig1}).

\begin{figure}
\begin{center}
	\psset{xunit=0.08cm,yunit=0.08cm,runit=0.08cm}
	\begin{pspicture}(0,0)(100,40)

 		\pscircle*(100,0){0.6}
 		\pscircle*(33,16){0.6}
 		\pscircle*(20,20){0.6}
 		\pscircle(20,20){15}
 		\put(101,-2){$y_i$}
 		\put(21,18){$z_C$}

		\psarc[linestyle=solid]{<->}(100,0){70}{147}{180}
		\psarc[linestyle=dashed]{<->}(100,0){84}{152}{180}

		\psline{->}(100,0)(40,36)
		\psline{->}(20,20)(8,11)
		\psline[linestyle=dashed]{->}(100,0)(16,3)

 		\rput{-29}(64,25){$opt_k$}
 		\rput{-2}(60,4.5){$opt_k+\radcost(\C_m)$}
 		\put(-20,20){$\radcost(\C_m)$}
		\psline[linewidth=0.1]{-}(3,20)(13,16)

	\end{pspicture}
\end{center}
\caption{Intermediate centers.}
\label{rad_2k_phasefig1}
\end{figure}
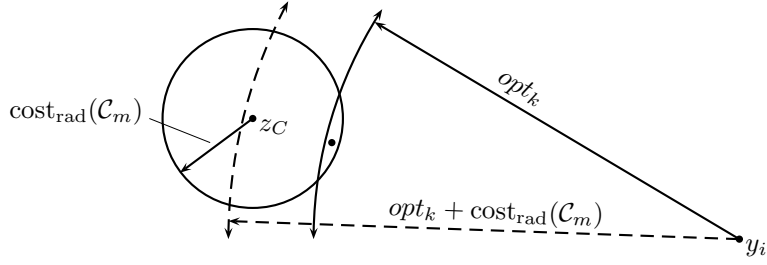

For \(\lambda\in\R\) with \(\lambda>0\) a ball of radius $\opt_k+\radcost(\C_m)$ can be covered by $\left\lceil\left(\frac{3}{\lambda}\right)^d\right\rceil$ balls of radius $\lambda\left(\opt_k+\radcost(\C_m)\right)$ (see \cite{naszodi}).
Choosing $\lambda=\frac{3}{\sqrt[d]{\left\lfloor\frac{m}{2k}\right\rfloor}}$, we get that each of the balls $\ball_{\opt_k+\radcost(\C_m)}(y_i)$ for $i=1,\ldots,k$ can be covered by \(\ell:=\left\lfloor\frac{m}{2k}\right\rfloor\leq\frac{m}{2k}\) balls of radius
\(\varepsilon=\frac{3}{\sqrt[d]{\left\lfloor\frac{m}{2k}\right\rfloor}}\left(\opt_k+\radcost(\C_m)\right)\).
Therefore, there exist \(k\cdot\ell\leq\frac{m}{2}\) balls $B_1,\ldots,B_{k\ell}$ of radius $\varepsilon$ such that each $z_C$ for $C\in\C_m$ is contained in at least one of these balls.
For $i=1,\ldots,k\ell$ let $a_i\in\R^d$ such that $B_i=\ball_\varepsilon(a_i)$.
Then, any cluster $C\in\C_m$ is contained in at least one of the balls $A_1,\ldots,A_{k\ell}$ with $A_i=\ball_{\radcost(\C_m)+\varepsilon}(a_i)$ for $i=1,\ldots,k\ell$ (see Figure~\ref{rad_2k_phasefig2}).
 
\begin{figure}[h]
\begin{center}
	\psset{xunit=0.08cm,yunit=0.08cm,runit=0.08cm}
	\begin{pspicture}(-30,0)(120,50)

 		\psarc*[linecolor=lightgray](45,5){45}{-7}{187}

 		\psarc*[linecolor=gray](45,5){30}{-10}{190}
 		\pscircle*(45,5){0.6}
 		\put(42,2){$a_i$}
		\psline{->}(45,5)(75,8)
 		\put(58,3){$\varepsilon$}
		\psline{->}(45,5)(30,47)
 		\put(33,42){$\radcost(\C_m)+\varepsilon$}

 		\pscircle*(20,15){0.6}
 		\pscircle(20,15){15}
		\psline{->}(20,15)(8,24)
 		\put(21,13){$z_{C_1}$}
		\psline[linewidth=0.1]{-}(-2,15)(13,19)
 		\put(-25,13){$\radcost(\C_m)$}

 		\pscircle*(60,25){0.6}
 		\pscircle(60,25){15}
 		\put(61,23){$z_{C_2}$}

		\psarc[linestyle=dashed]{<->}(60,-40){70}{45}{145}
		\psline[linestyle=dashed]{->}(79,0)(90,23)
 		\put(82,3){$opt_k+\radcost(\C_m)$}

	\end{pspicture}
\end{center}
\caption{Covering centers and clusters.}
\label{rad_2k_phasefig2}
\end{figure}
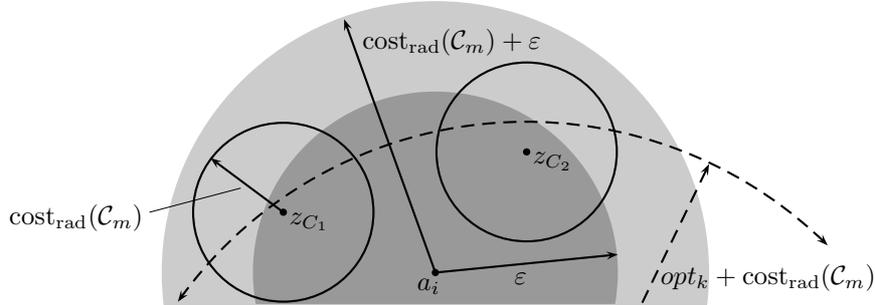

Let $t:=\left\lfloor\frac{3m}{4}\right\rfloor$ and $\C_m\cap \C_{t+1}$ be the set of clusters from $\C_m$ that still exist $\left\lceil\frac{m}{4}\right\rceil-1$ merge steps after the computation of $\C_m$.
In each iteration of its loop, the algorithm can merge at most two clusters from $\C_m$.
Thus, \(|\C_m\cap \C_{t+1}|>\frac{m}{2}\).

Since $k\ell\leq\frac{m}{2}$, there exist two clusters $C_1,C_2\in\C_m\cap \C_{t+1}$ that are contained in the same ball $A_i$ with $i\in\{1,\ldots,k\ell\}$.
Therefore, merging clusters $C_1$ and $C_2$ would result in a cluster whose radius can be upper bounded by $\rad(C_1\cup C_2)\leq\radcost(\C_m)+\varepsilon$.
Using Observation~\ref{rad_greedy_obs} and the fact that $C_1$ and $C_2$ are part of the clustering $\C_{t+1}$, we can upper bound the cost of $\C_t$ by
\[\radcost(\C_t)\leq\radcost(\C_m)+\varepsilon.\]

It remains to show \(\varepsilon<6\sqrt[d]{\frac{2k}{m}}\cdot\bigl(\opt_k+\radcost(\C_m)\bigr)\).
Since $\frac{m}{2k}>1$, we have $\frac{m}{2k}<2\left\lfloor\frac{m}{2k}\right\rfloor$.
Thus, $\sqrt[d]{\frac{m}{2k}}<\sqrt[d]{2\left\lfloor\frac{m}{2k}\right\rfloor}\leq2\sqrt[d]{\left\lfloor\frac{m}{2k}\right\rfloor}$ and $\frac{3}{\sqrt[d]{\left\lfloor\frac{m}{2k}\right\rfloor}}<6\sqrt[d]{\frac{2k}{m}}$.
\end{proof}

To prove Proposition~\ref{rad_2k_prop}, we apply Lemma~\ref{rad_2k_phaselem} for $\left\lceil\log_\frac{4}{3}\frac{|X|}{2k}\right\rceil$ consecutive phases.
\begin{proof}[Proof of Proposition~\ref{rad_2k_prop}]
Let \(u:=\left\lceil\log_\frac{4}{3}\frac{|X|}{2k}\right\rceil\) and define \(m_i:=\left\lceil\left(\frac{3}{4}\right)^i|X|\right\rceil\) for all $i=0,\ldots,u$.
Then, \(m_u\leq2k\) and \(m_i>2k\) for all $i=0,\ldots,u-1$.
Analogously to the proof of Proposition~\ref{drad_2k_prop}, we get \(\left\lfloor\frac{3m_i}{4}\right\rfloor\leq m_{i+1}\) and using Lemma~\ref{rad_2k_phaselem}, we deduce \(\radcost(\C_{m_{i+1}})<\left(1+6\sqrt[d]{\frac{2k}{m_i}}\right)\cdot\radcost(\C_{m_i})+6\sqrt[d]{\frac{2k}{m_i}}\cdot\opt_k\) for all $i=0,\ldots,u-1$.
By repeatedly applying this inequality and using \(\radcost(\C_{2k})\leq\radcost(\C_{m_u})\) and \(\radcost(\C_{m_0})=0\), we get
\begin{align}
\nonumber
\radcost\left(\C_{2k}\right)&<\sum_{i=0}^{u-1}\left(6\sqrt[d]{\frac{2k}{m_i}}\cdot\prod_{j=i+1}^{u-1}\left(1+6\sqrt[d]{\frac{2k}{m_j}}\right)\right)\cdot\opt_k\\
\label{rad_2k_prop_mi}
&\leq6\sqrt[d]{\frac{2k}{|X|}}\cdot\sum_{i=0}^{u-1}\left(\left(\frac{4}{3}\right)^\frac{i}{d}\cdot\prod_{j=i+1}^{u-1}\left(1+6\sqrt[d]{\frac{2k}{|X|}}\cdot\left(\frac{4}{3}\right)^\frac{j}{d}\right)\right)\cdot\opt_k\\
\label{rad_2k_prop_renum}
&=6\sqrt[d]{\frac{2k}{|X|}}\left(\frac{4}{3}\right)^\frac{u-1}{d}\cdot\sum_{i=0}^{u-1}\left(\left(\frac{3}{4}\right)^\frac{i}{d}\cdot\prod_{j=u-i}^{u-1}\left(1+6\sqrt[d]{\frac{2k}{|X|}}\cdot\left(\frac{4}{3}\right)^\frac{u-1}{d}\left(\frac{3}{4}\right)^\frac{u-1-j}{d}\right)\right)\cdot\opt_k.
\end{align}
Here, we obtain (\ref{rad_2k_prop_mi}) using $m_i\geq\left(\frac{3}{4}\right)^i|X|$ and we obtain (\ref{rad_2k_prop_renum}) by substituting $u-1-i$ for $i$.
Using $u-1<\log_\frac{4}{3}\frac{|X|}{2k}$, we deduce
\begin{equation}\label{rad_2k_propeq}
\radcost\left(\C_{2k}\right)<6\sum_{i=0}^{u-1}\left(\left(\frac{3}{4}\right)^\frac{i}{d}\cdot\prod_{j=0}^{i-1}\left(1+6\left(\frac{3}{4}\right)^\frac{j}{d}\right)\right)\cdot\opt_k.
\end{equation}
By taking only the first two terms of the series expansion of the exponential function, we get
\(1+6\left(\frac{3}{4}\right)^\frac{j}{d}<e^{6\left(\frac{3}{4}\right)^\frac{j}{d}}\) and therefore
\begin{equation}\label{rad_2k_prodeq}
\prod_{j=0}^{i-1}\left(1+6\left(\frac{3}{4}\right)^\frac{j}{d}\right)<\prod_{j=0}^{i-1}\left(e^{6\left(\frac{3}{4}\right)^\frac{j}{d}}\right)=e^{6\sum_{j=0}^{i-1}\left(\frac{3}{4}\right)^\frac{j}{d}}.
\end{equation}
The sum in the exponent can be bounded by the infinite geometric series
\begin{equation}\label{rad_2k_sumeq}
\sum_{j=0}^\infty\left(\frac{3}{4}\right)^\frac{j}{d}<\frac{1}{1-\left(\frac{3}{4}\right)^\frac{1}{d}}\leq4d,
\end{equation}
where the last inequality follows by upper bounding the convex function $f(x)=\left(\frac{3}{4}\right)^x$ in the interval $[0,1]$ by the line through $f(0)$ and $f(1)$.
Putting Inequalities (\ref{rad_2k_propeq}), (\ref{rad_2k_prodeq}) and (\ref{rad_2k_sumeq}) together then gives
\[\radcost\left(\C_{2k}\right)<6\sum_{i=0}^{u-1}\left(\left(\frac{3}{4}\right)^\frac{i}{d}\cdot e^{24d}\right)\cdot\opt_k<24d\cdot e^{24d}\cdot\opt_k,\]
where the last inequality follows by using Inequality~(\ref{rad_2k_sumeq}) again.
\end{proof}

\subsubsection{Connected instances}\label{connected_sec}
The analysis of the remaining merge steps from the discrete $k$-center case (cf. Section~\ref{drad_rem_sec}) is not transferable to the $k$-center case.
Again, as in the proof of Proposition~\ref{rad_2k_prop}, we are no longer able to derive a simple additive bound on the increase of the cost when merging two clusters.
In order to preserve the logarithmic dependency of the approximation factor on $k$, we show that it is sufficient to analyze Algorithm~\ref{rad_algo} on a subset $Y\subseteq X$ satisfying a certain connectivity property.
Using this property, we are able to apply a combinatorial approach that relies on the number of merge steps left.

We start by defining the connectivity property that will be used to relate clusters to an optimal $k$-clustering.
\begin{definition}
Let \(Z\subseteq\R^d\) and $r\in\R$.
Two sets $A,B\subseteq\R^d$ are called \emph{$(Z,r)$-connected} if there exists a $z\in Z$ with \(\ball_r^d(z)\cap A\neq\varnothing\) and \(\ball_r^d(z)\cap B\neq\varnothing\).
\end{definition}
Note that for any two $(Z,r)$-connected clusters $A,B$, we have
\begin{equation}\label{connect_eq}
\rad(A\cup B)\leq\rad(A)+2\cdot\rad(B)+2r.
\end{equation}

Next, we show that for any input set $X$ we can bound the cost of the $k$-clustering computed by Algorithm~\ref{rad_algo} by the cost of the $\ell$-clustering computed by the algorithm on a connected subset $Y\subseteq X$ for a proper $\ell\leq k$.
Recall that by our convention from the beginning of Section~\ref{rad_sec}, the clusterings computed by Algorithm~\ref{rad_algo} on a particular input set are uniquely determined.

\begin{lemma}\label{connect_lem}
Let \(X\subset\R^d\) be finite and $k\in\N$ with $k\leq|X|$.
Then, there exists a subset $Y\subseteq X$, a number $\ell\in\N$ with $\ell\leq k$ and $\ell\leq|Y|$, and a set \(Z\subset\R^d\) with \(|Z|=\ell\) such that:
\begin{enumerate}
\item\label{connect_property1}  $Y$ is $(\ell,\opt_k)$-coverable;
\item\label{connect_property2}  \(\radcost(\C_k)\leq\radcost(\cP_\ell)\);
\item\label{connect_property3} For all $n\in\N$ with $\ell+1\leq n\leq|Y|$, every cluster in $\cP_n$ is $(Z,\opt_k)$-connected to another cluster in $\cP_n$.
\end{enumerate}
Here, the collection $\cP_1,\ldots,\cP_{|Y|}$ denotes the hierarchical clustering computed by Algorithm~\ref{rad_algo} on input $Y$.
\end{lemma}
\begin{proof}
To define $Y,Z$, and $\ell$ we consider the $(k+1)$-clustering computed by Algorithm~\ref{rad_algo} on input $X$.
We know that \(X=\bigcup_{A\in \C_{k+1}}A\) is $(k,\opt_k)$-coverable.
Let \(E\subseteq\C_{k+1}\) be a minimal subset such that \(\bigcup_{A\in E}A\) is $(|E|-1,\opt_k)$-coverable, i.e., for all sets \(F\subseteq \C_{k+1}\) with $|F|<|E|$ the union \(\bigcup_{A\in F}A\) is not $(|F|-1,\opt_k)$-coverable.
Since a set $F$ of size $1$ cannot be $(|F|-1,\opt_k)$-coverable, we get \(|E|\geq2\).

Let $Y:=\bigcup_{A\in E}A$ and $\ell:=|E|-1$.
Then, $\ell\leq k$ and $Y$ is \((\ell,\opt_k)\)-coverable.
This establishes property \ref{connect_property1}.

It follows that there exists a set \(Z\subset\R^d\) with \(|Z|=\ell\) and \(Y\subset\bigcup_{z\in Z}\ball_{\opt_k}^d(z)\).
Furthermore, we let $\cP_1,\ldots,\cP_{|Y|}$ be the hierarchical clustering computed by Algorithm~\ref{rad_algo} on input $Y$.

Since $Y$ is the union of the clusters from $E\subseteq\C_{k+1}$, each merge step between the computation of $\C_{|X|}$ and $\C_{k+1}$ merges either two clusters $A,B\subset Y$ or two clusters $A,B\subset X\setminus Y$.
The merge steps inside $X\setminus Y$ have no influence on the clusters inside $Y$.
Furthermore, the merge steps inside $Y$ would be the same in the absence of the clusters inside $X\setminus Y$.
Therefore, on input $Y$, Algorithm~\ref{rad_algo} computes the $(\ell+1)$-clustering \(\cP_{\ell+1}=E=\C_{k+1}\cap2^Y\).
Thus, $\cP_{\ell+1}\subseteq\C_{k+1}$.

To compute $\cP_\ell$, on input $Y$, Algorithm~\ref{rad_algo} merges two clusters from $\cP_{\ell+1}$ that minimize the radius of the resulting cluster.
Analogously, on input $X$, Algorithm~\ref{rad_algo} merges two clusters from $\C_{k+1}$ to compute $\C_k$.
Since $\cP_{\ell+1}\subseteq\C_{k+1}$, Observation~\ref{rad_greedy_obs} implies \(\radcost(\C_k)\leq\radcost(\cP_{\ell})\), thus proving property \ref{connect_property2}.

It remains to show that for all $n\in\N$ with $\ell+1\leq n\leq|Y|$ it holds that every cluster in $\cP_n$ is $(Z,\opt_k)$-connected to another cluster in $\cP_n$ (property \ref{connect_property3}).
By the definition of $Z$, ever cluster in $\cP_n$ intersects at least one ball \(\ball_{\opt_k}^d(z)\) for $z\in Z$.
Therefore, it is enough to show that each ball \(\ball_{\opt_k}^d(z)\) intersects at least two clusters from $\cP_n$.
We first show this property for $n=\ell+1$.
For $\ell=1$ this follows from the fact that \(\ball_{\opt_k}^d(z)\) with \(Z=\{z\}\) has to contain both clusters from $\cP_2$.
For $\ell>1$, we are otherwise able to remove one cluster from $\cP_{\ell+1}$ and get $\ell$ clusters whose union is $(\ell-1,\opt_k)$-coverable.
This contradicts the definition of \(E=\cP_{\ell+1}\) as a minimal subset with this property.

To show property \ref{connect_property3} for general $n$, let $C_1\in\cP_n$ and $z\in Z$ with \(\ball_{\opt_k}^d(z)\cap C_1\neq\varnothing\).
There exists a unique cluster $\tilde C_1\in\cP_{\ell+1}$ with $C_1\subseteq\tilde C_1$.
Then, we have \(\ball_{\opt_k}^d(z)\cap \tilde C_1\neq\varnothing\).
However, we have just shown that \(\ball_{\opt_k}^{d}(z)\) has to intersect at least two clusters from $\cP_{\ell+1}$.
Thus, there exists another cluster $\tilde C_2\in\cP_{\ell+1}$ with \(\ball_{\opt_k}^d(z)\cap\tilde C_2\neq\varnothing\).
Since every cluster from $\cP_{\ell+1}$ is a union of clusters from $\cP_n$, there exists at least one cluster $C_2\in\cP_n$ with $C_2\subseteq\tilde C_2$ and \(\ball_{\opt_k}^d(z)\cap C_2\neq\varnothing\).
\end{proof}

\subsubsection{Analysis of the remaining merge steps}\label{rad_rem_sec}
Let $Y,Z,\ell$, and $\cP_1,\ldots,\cP_{|Y|}$ be as given by Lemma~\ref{connect_lem}.
Then, Proposition~\ref{rad_2k_prop} can be used to obtain an upper bound for the cost of $\cP_{2\ell}$.
In the following, we analyze the merge steps leading from $\cP_{2\ell}$ to $\cP_{\ell+1}$ and show how to obtain an upper bound for the cost of $\cP_{\ell+1}$.
As in Section~\ref{rad_2k_sec}, we analyze the merge steps in phases.
The following lemma is used to bound the increase of the cost during a single phase.
Note that $\opt_k$ still refers to the cost of an optimal solution on input $X$, not $Y$.

\begin{lemma}\label{rad_rem_phaselem}
Let $m,n\in\N$ with \(n\leq2\ell\) and \(\ell<m\leq n\leq|Y|\).
If there are no two $(Z,\opt_k)$-connected clusters in $\cP_m\cap\cP_n$, it holds
\[\radcost(\cP_{\left\lfloor\frac{m+\ell}{2}\right\rfloor})\leq\radcost(\cP_m)+2\cdot\radcost(\cP_n)+2\opt_k.\]
\end{lemma}
\begin{proof}
We show that there exist at least $m-\ell$ disjoint pairs of clusters from $\cP_m$ such that the radius of their union can be upper bounded by \(\radcost(\cP_m)+2\cdot\radcost(\cP_n)+2\opt_k\).
By Observation~\ref{rad_greedy_obs}, this upper bounds the cost of the computed clusterings as long as such a pair of clusters remains.
Then, the lemma follows from the fact that in each iteration of its loop the algorithm can destroy at most two of these pairs.
To bound the number of these pairs of clusters, we start with a structural observation.
$\cP_m\cap \cP_n$ is the set of clusters from $\cP_n$ that still exist in $\cP_m$.
By our definition of $Y,Z$, and $\ell$, we conclude that any cluster $A\in\cP_m\cap\cP_n$ is $(Z,\opt_k)$-connected to another cluster \(B\in \cP_m\).
If we assume that there are no two $(Z,\opt_k)$-connected clusters in $\cP_m\cap\cP_n$, this implies \(B\in \cP_m\setminus\cP_n\) (see Figure~\ref{rad_rem_fig}).
Thus, using \(A\in\cP_n\), \(B\in\cP_m\), and Inequality~(\ref{connect_eq}), the radius of $A\cup B$ can be bounded by
\begin{equation}\label{rad_rem_phasebound1}
\rad(A\cup B)\leq\radcost(\cP_m)+2\cdot\radcost(\cP_n)+2\opt_k.
\end{equation}
Moreover, using a similar argument, we derive the same bound for two clusters $A_1,A_2\in\cP_m\cap\cP_n$ that are $(Z,\opt_k)$-connected to the same cluster \(B\in \cP_m\setminus\cP_n\).
That is,
\begin{equation}\label{rad_rem_phasebound2}
\rad(A_1\cup A_2)\leq\radcost(\cP_m)+2\cdot\radcost(\cP_n)+2\opt_k.
\end{equation}

\begin{figure}
\begin{center}
	\psset{xunit=0.08cm,yunit=0.08cm,runit=0.08cm}
	\begin{pspicture}(0,0)(180,30)

 		\psarc*[linecolor=gray](50,5){25}{-12}{192}
		\psline{->}(50,5)(40,28)
 		\put(46,17){$\opt_k$}

 		\psarc*[linecolor=gray](125,-5){25}{11}{169}

		\psarc{<->}(90,0){20}{0}{180}
		\psline{->}(90,0)(76,14)
 		\put(81,10){$\radcost(\cP_m)$}
 		\put(100,19){\Large$B$}

		\pscircle(22,18){10}
		\pscircle(140,20){10}
		\psline{->}(140,20)(147,27)
		\psline[linewidth=0.1]{-}(144,22)(153,20)
 		\put(154,19){$\radcost(\cP_n)$}
 		\put(10,6){\Large$A_1$}
 		\put(148,11){\Large$A_2$}

		\pscircle*(73,6){0.6}
		\pscircle*(105,3){0.6}
		\pscircle*(29,15){0.6}
		\pscircle*(140,13){0.6}

	\end{pspicture}
\end{center}
\caption{Merging $(Z,\opt_k)$-connected clusters.}
\label{rad_rem_fig}
\end{figure}
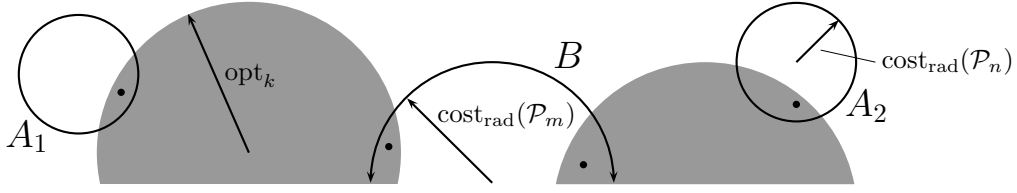

Next, we show that there exist at least $\left\lceil\frac{|\cP_m\cap\cP_n|}{2}\right\rceil$ disjoint pairs of clusters from $\cP_m$ such that the radius of their union can be bounded either by Inequality~(\ref{rad_rem_phasebound1}) or by Inequality~(\ref{rad_rem_phasebound2}).
To do so, we first consider the pairs of clusters from $\cP_m\cap\cP_n$ that are $(Z,\opt_k)$-connected to the same cluster from \(\cP_m\setminus\cP_n\) until no candidates are left.
For these pairs, we can bound the radius of their union by Inequality~(\ref{rad_rem_phasebound2}).
Then, each cluster from \(\cP_m\setminus\cP_n\) is $(Z,\opt_k)$-connected to at most one of the remaining clusters from $\cP_m\cap\cP_n$. 
Thus, each remaining cluster $A\in\cP_m\cap\cP_n$ can be paired with a different cluster \(B\in \cP_m\setminus\cP_n\) such that $A$ and $B$ are $(Z,\opt_k)$-connected.
For these pairs, we can bound the radius of their union by Inequality~(\ref{rad_rem_phasebound1}).
Since for all pairs either one or both of the clusters come from the set $\cP_m\cap\cP_n$, we can lower bound the number of pairs by $\left\lceil\frac{|\cP_m\cap\cP_n|}{2}\right\rceil$.

To complete the proof, we show that \(m-\ell\leq\left\lceil\frac{|\cP_m\cap\cP_n|}{2}\right\rceil\).
In each iteration of its loop, the algorithm can merge at most two clusters from $\cP_n$.
Therefore, there are at least $\left\lceil\frac{n-|\cP_m\cap\cP_n|}{2}\right\rceil$ merge steps between the computations of $\cP_n$ and $\cP_m$.
Hence, \(m\leq n-\left\lceil\frac{n-|\cP_m\cap\cP_n|}{2}\right\rceil\leq\frac{n}{2}+\frac{|\cP_m\cap\cP_n|}{2}\).
Using \(n\leq2\ell\), we get \(m-\ell\leq\frac{|\cP_m\cap\cP_n|}{2}\).
\end{proof}

\begin{lemma}\label{rad_rem_lem}
Let $n\in\N$ with \(n\leq2\ell\) and \(\ell<n\leq|Y|\).
Then,
\[\radcost(\cP_{\ell+1})<2(\log_2(\ell)+2)\cdot\left(\radcost(\cP_n)+\opt_k\right).\]
\end{lemma}
\begin{proof}
For $n=\ell+1$ there is nothing to show.
Hence, assume $n>\ell+1$.
Then, by definition of $Z$, there exist two $(Z,\opt_k)$-connected clusters in $\cP_n$.
Now let $\tilde{n}\in\N$ with $\tilde{n}<n$ be maximal such that no two $(Z,\opt_k)$-connected clusters exist in $\cP_{\tilde{n}}\cap \cP_n$.
The number $\tilde{n}$ is well-defined since $|\cP_1|=1$ implies $\tilde{n}\geq1$.
It follows that the same holds for all $m\in\N$ with $m\leq\tilde{n}$.
We conclude that Lemma~\ref{rad_rem_phaselem} is applicable for all $m\in\N$ with $\ell<m\leq\tilde{n}$.

By the definition of $\tilde{n}$ there still exist at least two $(Z,\opt_k)$-connected clusters in $\cP_{\tilde{n}+1}\cap\cP_n$.
Then, Observation~\ref{rad_greedy_obs} implies
\begin{equation}\label{rad_rem_lemeq}
\radcost(\cP_{\tilde{n}})\leq2\cdot\radcost(\cP_n)+\opt_k.
\end{equation}
If $\tilde{n}\leq\ell+1$ then Inequality~(\ref{rad_rem_lemeq}) proves the lemma.
For $\tilde{n}>\ell+1$ let \(u:=\left\lceil\log_2(\tilde{n}-\ell)\right\rceil\) and define \(m_i:=\left\lceil\left(\frac{1}{2}\right)^i(\tilde{n}-\ell)+\ell\right\rceil>\ell\) for all $i=0,\ldots,u$.
Then, \(m_0=\tilde{n}\) and \(m_u=\ell+1\).
Furthermore, we obtain
\begin{align*}
\textstyle\left\lfloor\frac{m_i+\ell}{2}\right\rfloor
&\textstyle=\left\lfloor\frac{1}{2}\left\lceil\left(\frac{1}{2}\right)^i(\tilde{n}-\ell)+\ell\right\rceil+\frac{\ell}{2}\right\rfloor\leq\left\lfloor\frac{1}{2}\left(\left(\frac{1}{2}\right)^i(\tilde{n}-\ell)+\ell+1\right)+\frac{\ell}{2}\right\rfloor\\
&\textstyle=\left\lfloor\left(\frac{1}{2}\right)^{i+1}(\tilde{n}-\ell)+\ell+\frac{1}{2}\right\rfloor\leq\left\lceil\left(\frac{1}{2}\right)^{i+1}(\tilde{n}-\ell)+\ell\right\rceil=m_{i+1}.
\end{align*}
Since Algorithm~\ref{rad_algo} uses a greedy strategy, we deduce \(\radcost(\cP_{m_{i+1}})\leq\radcost(\cP_{\left\lfloor\frac{m_i+\ell}{2}\right\rfloor})\) for all $i=0,\ldots,u-1$.
Combining this with Lemma~\ref{rad_rem_phaselem} (applied to $m=m_i$), we obtain
\[\radcost(\cP_{m_{i+1}})\leq\radcost(\cP_{m_i})+2\cdot\radcost(\cP_n)+2\opt_k.\]
By repeatedly applying this inequality for $i=0,\ldots,u-1$ and summing up the costs, we get
\[\radcost(\cP_{m_u})<\radcost(\cP_{\tilde{n}})+2u\cdot\left(\radcost(\cP_n)+\opt_k\right)\stackrel{(\ref{rad_rem_lemeq})}{<}2(u+1)\cdot\left(\radcost(\cP_n)+\opt_k\right).\]
Since $\tilde{n}<2\ell$, we get \(u<\log_2(\ell)+1\) and the lemma follows using \(m_u=\ell+1\).
\end{proof}

The following lemma finishes the analysis except for the last merge step.
\begin{lemma}\label{rad_allstages}
Let \(Y\subset\R^d\) be finite and $\ell\leq|Y|$ such that $Y$ is $(\ell,\opt_k)$-coverable.
Furthermore, let \(Z\subset\R^d\) with $|Z|=\ell$ such that for all $n\in\N$ with $\ell+1\leq n\leq|Y|$ every cluster in $\cP_n$ is $(Z,\opt_k)$-connected to another cluster in $\cP_n$,  where $\cP_1,\ldots,\cP_{|Y|}$ denotes the hierarchical clustering computed by Algorithm~\ref{rad_algo} on input $Y$.
Then,
\[\radcost(\cP_{\ell+1})<2(\log_2(\ell)+2)\cdot\left(24d\cdot e^{24d}+1\right)\cdot\opt_k.\]
\end{lemma}
\begin{proof}
Let $n:=\min(|Y|, 2\ell)$.
Then, using Proposition~\ref{rad_2k_prop}, we get \(\radcost(\cP_n)<24d\cdot e^{24d}\cdot\opt_k\).
The lemma follows by using this bound in combination with Lemma~\ref{rad_rem_lem}.
\end{proof}

\subsubsection{Proof of Theorem~\ref{rad_result}}\label{rad_proof}
Using Lemma~\ref{connect_lem}, we know that there is a subset $Y\subseteq X$, a number $\ell\leq k$, and a hierarchical clustering $\cP_1,\ldots,\cP_{|Y|}$ of $Y$ with \(\radcost(\C_k)\leq\radcost(\cP_{\ell})\).
Furthermore, there is a set $Z\subset\R^d$ such that every cluster from $\cP_{\ell+1}$ is $(Z,\opt_k)$-connected to another cluster in $\cP_{\ell+1}$.
Thus, $\cP_{\ell+1}$ contains two clusters $A,B$ that intersect with the same ball of radius $\opt_k$.
Hence
\[\radcost(\C_k)\leq\rad(A\cup B)\leq 2\cdot\radcost(\cP_{\ell+1})+\opt_k.\]
The theorem follows using Lemma~\ref{rad_allstages} and $\ell\leq k$.
\qed

\subsection{Diameter $k$-clustering}\label{diam_sec}
In this section, we analyze the agglomerative complete linkage clustering algorithm for Problem~\ref{diam_prob} stated as Algorithm~\ref{diam_algo}.
\begin{table}[t]
\centering
\begin{minipage}{.7\textwidth}
\footnotesize
\hrule
\textsc{AgglomerativeCompleteLinkage}$(X)$:\\
\begin{tabular}{rl}
$X$ & finite set of input points from $\R^d$\\
\end{tabular}
\hrule
\begin{algorithmic}[1]
\State $\C_{|X|}:=\left\{\,\{x\}\,|\,x\in X\right\}$
\For{$i=|X|-1,\ldots,1$}
\State find distinct clusters $A,B\in \C_{i+1}$ minimizing \(\diam(A\cup B)\)\label{diam_algo_min}
\State $\C_i:=(\C_{i+1}\setminus\{A,B\})\,\cup\,\{A\cup B\}$
\EndFor
\State\Return $\C_1,\ldots,\C_{|X|}$
\end{algorithmic}
\hrule
\end{minipage}
\vspace{1em}
\caption{The agglomerative complete linkage clustering algorithm.}
\label{diam_algo}
\end{table}
Again, the only difference to Algorithm~\ref{drad_algo} and \ref{rad_algo} is the minimization of the diameter in Step~\ref{diam_algo_min}.
As in the analysis of Algorithm~\ref{rad_algo}, we may assume that for any input set $X$ the clusterings computed by Algorithm~\ref{diam_algo} are uniquely determined, i.e. the minimum in Step~\ref{diam_algo_min} is always unambiguous.

Note that in this section \emph{cost} always means diameter cost and $\opt_k$ refers to the cost of an optimal diameter $k$-clustering of $X\subset\R^d$ where $k\in\N$ with \(k\leq|X|\).
Analogously to the (discrete) radius case, any cluster $C$ is contained in a ball of radius $\diam(C)$ and thus the set $X$ is $(k,\opt_k)$-coverable.

\begin{observation}[analogous to Observation~\ref{drad_greedy_obs} and \ref{rad_greedy_obs}]\label{diam_greedy_obs}
The cost of all computed clusterings is equal to the diameter of the cluster created last.
Furthermore, the diameter of the union of any two clusters is always an upper bound for the cost of the clustering to be computed next.
\end{observation}

The following theorem states our main result.
\begin{theorem}\label{diam_result}
Let \(X\subset\R^d\) be a finite set of points.
Then, for all $k\in\N$ with $k\leq|X|$, the partition $\C_k$ of $X$ into $k$ clusters as computed by Algorithm~\ref{diam_algo} satisfies
\[\diamcost(\C_k)=O(\log k)\cdot\opt_k,\]
where $\opt_k$ denotes the cost of an optimal solution to Problem~\ref{diam_prob}, and the constant hidden in the $O$-notation is doubly exponential in the dimension $d$.
\end{theorem}

As in the proof of Theorem~\ref{drad_result} and \ref{rad_result}, we first show a bound for the cost of the intermediate $2k$-clustering.
However, we have to apply a different analysis again.
This time, the new analysis results in a bound that depends doubly exponential on the dimension.

\subsubsection{Analysis of the $2k$-clustering}\label{diam_2k_sec}
\begin{proposition}\label{diam_2k_prop}
Let \(X\subset\R^d\) be finite.
Then, for all $k\in\N$ with $2k\leq|X|$, the partition $\C_{2k}$ of $X$ into $2k$ clusters as computed by Algorithm~\ref{diam_algo} satisfies
\[\diamcost(\C_{2k})<2^{3\sigma}\left(28d+6\right)\cdot\opt_k,\]
where \(\sigma=(42d)^d\) and $\opt_k$ denotes the cost of an optimal solution to Problem~\ref{diam_prob}.
\end{proposition}

In our analysis of the $k$-center problem, we made use of the fact that merging two clusters lying inside a ball of some radius $r$ results in a new cluster of radius at most $r$.
This is no longer true for the diameter $k$-clustering problem.
We are not able to derive a bound for the diameter of the new cluster that is significantly less than $2r$.
The additional factor of $2$ makes our analysis from Section~\ref{rad_2k_sec} useless for the diameter case.

To prove Proposition~\ref{diam_2k_prop}, we divide the merge steps of Algorithm~\ref{diam_algo} into two stages.
The first stage consists of the merge steps down to a $2^{2^{O(d\log d)}}k$-clustering.
The analysis of the first stage is based on the following notion of similarity.
Two clusters are called similar if one cluster can be translated such that every point of the translated cluster is near a point of the second cluster.
Then, by merging similar clusters, the diameter essentially increases by the length of the translation vector.
During the first stage, we guarantee that there is a sufficiently large number of similar clusters left.
The cost of the intermediate $2^{2^{O(d\log d)}}k$-clustering can be upper bounded by $O(d)\cdot\opt_k$.

The second stage consists of the steps reducing the number of remaining clusters from $2^{2^{O(d\log d)}}k$ to only $2k$.
In this stage, we are no longer able to guarantee that a sufficiently large number of similar clusters exists.
Therefore, we analyze the merge steps of the second stage using a weaker argument, very similar to the one used in the second step of the analysis in the discrete $k$-center case (cf. Section~\ref{drad_rem_sec}).
As long as there are more than $2k$ clusters left, we are able to find sufficiently many pairs of clusters that intersect with the same cluster of an optimal $k$-clustering.
Therefore, we can bound the cost of merging such a pair by the sum of the diameters of the two clusters plus the diameter of the optimal cluster.
We find that the cost of the intermediate $2k$-clustering is upper bounded by $2^{2^{O(d\log d)}}\cdot\opt_k$.
Let us remark that we do not obtain our main result if we already use this argument for the first stage.

Both stages are again subdivided into phases, such that in each phase the number of remaining clusters is reduced by one fourth.

\subsubsection{Stage one}
The following lemma will be used to bound the increase of the cost during a single phase.

\begin{lemma}\label{diam_stage1_phaselem}
Let \(\lambda\in\R\) with \(0<\lambda<1\) and \(\rho:=\left\lceil\left(\frac{3}{\lambda}\right)^d\right\rceil\).
Furthermore, let $m\in\N$ with \(2^{\rho+1}k<m\leq|X|\).
Then,
\begin{equation}\label{diam_stage1_eq}
\diamcost(\C_{\left\lfloor\frac{3m}{4}\right\rfloor})<\left(1+2\lambda\right)\cdot\diamcost(\C_m)+4\sqrt[d]{\frac{2^{\rho+1}k}{m}}\cdot\opt_k.
\end{equation}
\end{lemma}
\begin{proof}
From every cluster $C\in\C_m$, we fix an arbitrary point and denote it by $p_C$.
Let $R:=\diamcost(\C_m)$.
Then, the distance from $p_C$ to any $q\in C$ is at most $R$ and we get \(C-p_C\subset\ball_R^d(0)\).

A ball of radius $R$ can be covered by $\rho$ balls of radius $\lambda R$ (see \cite{naszodi}).
Hence, there exist \(y_1,\ldots,y_\rho\in\R^d\) with \(\ball_R^d(0)\subseteq\bigcup_{i=1}^\rho\ball_{\lambda R}^d(y_i)\).
For $C\in\C_m$, we call the set
\(\conf(C):=\{y_i\ |\ 1\leq i\leq\rho\text{ and }\ball_{\lambda R}^d(y_i)\cap(C-p_C)\neq\varnothing\}\)
the configuration of $C$.
That is, we identify each cluster $C\in\C_m$ with the subset of the balls \(\ball_{\lambda R}^d(y_1),\ldots,\ball_{\lambda R}^d(y_\rho)\) that intersect with $C-p_C$.
Note that no cluster from $C\in\C_m$ has an empty configuration.
The number of possible configurations is upper bounded by \(2^\rho\).

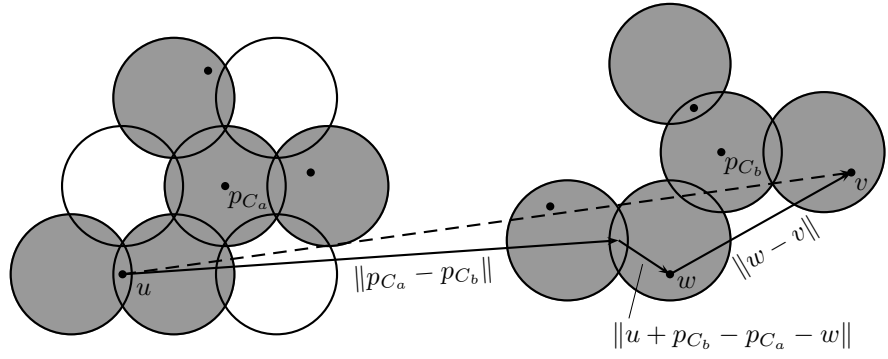
\begin{figure}
\begin{center}
	\psset{xunit=0.09cm,yunit=0.09cm,runit=0.09cm}
	\begin{pspicture}(0,0)(125,50)
	
		\pscircle*[linecolor=gray](10,10){9}
		\pscircle*[linecolor=gray](25,10){9}
		\pscircle*[linecolor=gray](32.5,23){9}
		\pscircle*[linecolor=gray](47.5,23){9}
		\pscircle*[linecolor=gray](25,36){9}

		\pscircle(10,10){9}
		\pscircle(25,10){9}
		\pscircle(40,10){9}
		\pscircle(17.5,23){9}
		\pscircle(32.5,23){9}
		\pscircle(47.5,23){9}
		\pscircle(25,36){9}
		\pscircle(40,36){9}

		\pscircle*(32.5,23){0.6}
 		\put(33,20.5){$p_{C_a}$}
		\pscircle*(17.5,10){0.6}
 		\put(19.5,7){$u$}
		\pscircle*(30,40){0.6}
		\pscircle*(45,25){0.6}
		
		\pscircle*[linecolor=gray](82.5,15){9}
		\pscircle*[linecolor=gray](97.5,15){9}
		\pscircle*[linecolor=gray](105,28){9}
		\pscircle*[linecolor=gray](120,28){9}
		\pscircle*[linecolor=gray](97.5,41){9}

		\pscircle(82.5,15){9}
		\pscircle(97.5,15){9}
		\pscircle(105,28){9}
		\pscircle(120,28){9}
		\pscircle(97.5,41){9}

		\pscircle*(105,28){0.6}
 		\put(105.5,25.5){$p_{C_b}$}
		\pscircle*(80,20){0.6}
		\pscircle*(97.5,10){0.6}
 		\put(98.5,8){$w$}
		\pscircle*(101,34.5){0.6}
		\pscircle*(124,25){0.6}
 		\put(124.5,22.5){$v$}
		
		\psline{->}(17.5,10)(90,15)
 		\rput{4}(61.5,10){$\|p_{C_a}-p_{C_b}\|$}
		\psline{->}(90,15)(97.5,10)
		\psline[linewidth=0.1]{-}(93.5,11.5)(91.5,3.5)
 		\put(89,0){$\|u+p_{C_b}-p_{C_a}-w\|$}
		\psline{->}(97.5,10)(124,25)
 		\rput{30}(113,14){$\|w-v\|$}
		\psline[linestyle=dashed]{->}(17.5,10)(124,25)

	\end{pspicture}
\end{center}
\caption{Congruent configurations.}
\label{diam_stage1_fig}
\end{figure}

Let $t:=\left\lfloor\frac{3m}{4}\right\rfloor$ and $\C_m\cap \C_{t+1}$ be the set of clusters from $\C_m$ that still exist $\left\lceil\frac{m}{4}\right\rceil-1$ merge steps after the computation of $\C_m$.
In each iteration of its loop, the algorithm can merge at most two clusters from $\C_m$.
Thus, \(|\C_m\cap \C_{t+1}|>\frac{m}{2}\).
It follows that there exist $j>\frac{m}{2^{\rho+1}}$ distinct clusters $C_1,\ldots,C_j\in\C_m\cap \C_{t+1}$ with the same configuration.
Using $m>2^{\rho+1}k$, we deduce $j>k$.

Let \(P:=\{p_{C_1},\ldots,p_{C_j}\}\).
Since $X$ is $(k,\opt_k)$-coverable, so is \(P\subset X\).
Therefore, by Lemma~\ref{volume_lem}, there exist distinct $a,b\in\{1,\ldots,j\}$ such that
\(\|p_{C_a}-p_{C_b}\|\leq4\sqrt[d]{\frac{2^{\rho+1}k}{m}}\cdot\opt_k\).

Next, we want to bound the diameter of the union of the corresponding clusters $C_a$ and $C_b$.
The distance between any two points $u,v\in C_a$ or $u,v\in C_b$ is at most the cost of $\C_m$.
Now let $u\in C_a$ and $v\in C_b$.
Using the triangle inequality, for any $w\in\R^d$, we obtain
\(\|u-v\|\leq\|p_{C_a}-p_{C_b}\|+\|u+p_{C_b}-p_{C_a}-w\|+\|w-v\|\) (see Figure~\ref{diam_stage1_fig}).

For \(\|p_{C_a}-p_{C_b}\|\), we just derived an upper bound.
To bound $\|u+p_{C_b}-p_{C_a}-w\|$, we let $y\in\conf(C_a)=\conf(C_b)$ such that $u-p_{C_a}\in\ball_{\lambda R}^d(y)$.
Furthermore, we fix $w\in C_b$ with \(w-p_{C_b}\in\ball_{\lambda R}^d(y)\).
Hence, $\|u+p_{C_b}-p_{C_a}-w\|=\|u-p_{C_a}-(w-p_{C_b})\|$ can be upper bounded by \(2\lambda R=2\lambda\cdot\diamcost(\C_m)\).
For $w\in C_b$ the distance $\|w-v\|$ is bounded by $\diam(C_b)\leq \diamcost(\C_m)$.
We conclude that merging clusters $C_a$ and $C_b$ results in a cluster whose diameter can be upper bounded by
\[\diam(C_a\cup C_b)<\left(1+2\lambda\right)\cdot\diamcost(\C_m)+4\sqrt[d]{\frac{2^{\rho+1}k}{m}}\cdot\opt_k.\]
Using Observation~\ref{diam_greedy_obs} and the fact that $C_a$ and $C_b$ are part of the clustering $\C_{t+1}$, we can upper bound the cost of $\C_t$ by
\(\diamcost(\C_t)\leq\diam(C_a\cup C_b)\).
\end{proof}

Note that the parameter $\lambda$ from Lemma~\ref{diam_stage1_phaselem} establishes a trade-off between the two terms on the right-hand side of Inequality~(\ref{diam_stage1_eq}).
To complete the analysis of the first stage, we have to carefully choose $\lambda$.
In the proof of the following lemma, we use \(\lambda=\nicefrac{\ln\frac{4}{3}}{4d}\) and apply Lemma~\ref{diam_stage1_phaselem} for $\left\lceil\log_\frac{4}{3}\frac{|X|}{2^{\sigma+1}k}\right\rceil$ consecutive phases, where \(\sigma=(42d)^d\).
Then, we are able to upper bound the total increase of the cost by a term that is linear in $d$ and $r$ and independent of $|X|$ and $k$.
The number of remaining clusters is independent of the number of input points $|X|$ and only depends on the dimension $d$ and the desired number of clusters $k$.
\begin{lemma}\label{diam_stage1_lem}
Let \(2^{\sigma+1}k<|X|\) for \(\sigma=(42d)^d\).
Then, on input $X$, Algorithm~\ref{diam_algo} computes a clustering \(\C_{2^{\sigma+1}k}\) with
\(\diamcost\left(\C_{2^{\sigma+1}k}\right)<\left(28d+4\right)\cdot\opt_k\).
\end{lemma}
\begin{proof}
Let \(u:=\left\lceil\log_\frac{3}{4}\frac{2^{\sigma+1}k}{|X|}\right\rceil\) and define \(m_i:=\left\lceil\left(\frac{3}{4}\right)^i|X|\right\rceil\) for all $i=0,\ldots,u$.
Furthermore, let \(\lambda=\nicefrac{\ln\frac{4}{3}}{4d}\).
This implies \(\rho\leq\sigma\) for the parameter $\rho$ of Lemma~\ref{diam_stage1_phaselem}.
Then, \(m_u\leq2^{\sigma+1}k\) and \(m_i>2^{\sigma+1}k\geq2^{\rho+1}k\) for all $i=0,\ldots,u-1$.
Since \(\left\lfloor\frac{3m_i}{4}\right\rfloor=\left\lfloor\frac{3}{4}\left\lceil\left(\frac{3}{4}\right)^i|X|\right\rceil\right\rfloor\leq\left\lfloor\left(\frac{3}{4}\right)^{i+1}|X|+\frac{3}{4}\right\rfloor\leq\left\lceil\left(\frac{3}{4}\right)^{i+1}|X|\right\rceil=m_{i+1}\) and Algorithm~\ref{diam_algo} uses a greedy strategy, we deduce \(\diamcost(\C_{m_{i+1}})\leq\diamcost(\C_{\left\lfloor\frac{3m_i}{4}\right\rfloor})\) for all $i=0,\ldots,u-1$.
Combining this with Lemma~\ref{diam_stage1_phaselem} (applied to $m=m_i$), we obtain
\[\diamcost(\C_{m_{i+1}})<\left(1+2\lambda\right)\cdot\diamcost(\C_{m_i})+4\sqrt[d]{\frac{2^{\rho+1}k}{m_i}}\cdot\opt_k.\]
By repeatedly applying this inequality for $i=0,\ldots,u-1$ and using \(\diamcost(\C_{2^{\sigma+1}k})\leq\diamcost(\C_{m_u})\) and \(\diamcost(\C_{m_0})=0\), we get
\begin{align*}
\diamcost\left(\C_{2^{\sigma+1}k}\right)&<\sum_{i=0}^{u-1}\left(\left(1+2\lambda\right)^i\cdot4\sqrt[d]{\frac{2^{\sigma+1}k}{\left(\frac{3}{4}\right)^{u-1-i}|X|}}\opt_k\right)\\
&=4\sqrt[d]{\frac{2^{\sigma+1}k}{\left(\frac{3}{4}\right)^{u-1}|X|}}\opt_k\cdot\sum_{i=0}^{u-1}\left(\left(1+2\lambda\right)^i\cdot\sqrt[d]{\left(\frac{3}{4}\right)^{i}}\right).
\end{align*}
Using $u-1<\log_\frac{3}{4}\frac{2^{\sigma+1}k}{|X|}$, we deduce
\begin{equation}\label{diam_stage1_lemeq}
\diamcost\left(\C_{2^{\sigma+1}k}\right)<4\opt_k\cdot\sum_{i=0}^{u-1}\left(\frac{1+2\lambda}{\sqrt[d]{\frac{4}{3}}}\right)^i.
\end{equation}
By taking only the first two terms of the series expansion of the exponential function, we get
\(1+2\lambda=1+\frac{\ln\frac{4}{3}}{2d}<e^\frac{\ln\frac{4}{3}}{2d}=\sqrt[2d]{\frac{4}{3}}\).
Substituting this bound into Inequality~(\ref{diam_stage1_lemeq}) and extending the sum gives
\[\diamcost\left(\C_{2^{\sigma+1}k}\right)<4\opt_k\cdot\sum_{i=0}^\infty\left(\frac{1}{\sqrt[2d]{\frac{4}{3}}}\right)^i<4\opt_k\cdot\sum_{i=0}^\infty\left(\frac{1}{1+2\lambda}\right)^i.\]
Solving the geometric series leads to
\[\diamcost\left(\C_{2^{\sigma+1}k}\right)<4\left(\frac{1}{2\lambda}+1\right)\cdot\opt_k<\left(28d+4\right)\cdot\opt_k.\]
\end{proof}

\subsubsection{Stage two}
The second stage covers the remaining merge steps until Algorithm~\ref{diam_algo} computes the clustering $\C_{2k}$.
However, compared to stage one, the analysis of a single phase yields a weaker bound.
The following lemma provides an analysis of a single phase of the second stage.
It is very similar to Lemma~\ref{drad_2k_phaselem} and Lemma~\ref{drad_rem_phaselem} in the analysis of the discrete $k$-center problem.
\begin{lemma}\label{diam_stage2_phaselem}
Let $m\in\N$ with \(2k<m\leq|X|\).
Then,
\[\diamcost(\C_{\left\lfloor\frac{3m}{4}\right\rfloor})<2\cdot\left(\diamcost(\C_m)+\opt_k\right).\]
\end{lemma}
\begin{proof}
Let $t:=\left\lfloor\frac{3m}{4}\right\rfloor$.
Then, $\C_m\cap \C_{t+1}$ is the set of clusters from $\C_m$ which still exist $\left\lceil\frac{m}{4}\right\rceil-1<\frac{m}{4}$ merge steps after the computation of $\C_m$.
In each iteration of its loop the algorithm can merge at most two clusters from $\C_m$.
Thus, \(|\C_m\cap \C_{t+1}|>\frac{m}{2}>k\).
Since $X$ is $(k,\opt_k)$-coverable there exists a point $y\in\R^d$ such that $\ball_{\opt_k}^d(y)$ intersects with two clusters $A,B\in\C_m\cap \C_{t+1}$.
We conclude that merging $A$ and $B$ would result in a cluster whose diameter can be upper bounded by
\(\diam(A\cup B)<2\diamcost(\C_m)+2\opt_k\) (cf. Figure~\ref{diam_stage2_fig}).
The result follows using $A,B\in\C_{t+1}$ and Observation~\ref{diam_greedy_obs}.
\end{proof}

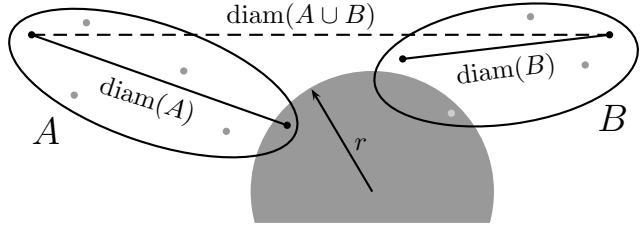
\begin{figure}
\begin{center}
	\psset{xunit=0.08cm,yunit=0.08cm,runit=0.08cm}
	\begin{pspicture}(0,0)(120,35)

 		\psarc*[linecolor=gray](61,5){20}{-15}{195}
		\psline{->}(61,5)(51,22)
 		\put(58,12){$r$}

		\rput{-19}(25,23){\psellipse(0,0)(25,10)}
		\rput{8}(83,26){\psellipse(0,0)(22,10)}
 		\put(5,13){\Large$A$}
 		\put(98,15){\Large$B$}

		\pscircle*[linecolor=gray](14,33){0.6}
		\pscircle*[linecolor=gray](30,25){0.6}
		\pscircle*[linecolor=gray](37,15){0.6}
		\pscircle*[linecolor=gray](12,21){0.6}
		\pscircle*(5,31){0.6}
		\pscircle*(47,16){0.6}
		\psline(5,31)(47,16)
 		\rput{-18}(24,20){$\diam(A)$}

		\pscircle*[linecolor=lightgray](74,18){0.6}
		\pscircle*[linecolor=gray](87,34){0.6}
		\pscircle*[linecolor=gray](96,26){0.6}
		\pscircle*(66,27){0.6}
		\pscircle*(100,31){0.6}
		\psline(66,27)(100,31)
 		\rput{7}(83,25){$\diam(B)$}

		\psline[linestyle=dashed](5,31)(100,31)
 		\put(37,33){$\diam(A\cup B)$}
		
	\end{pspicture}
\end{center}
\caption{Merging two clusters intersecting with a ball of radius $r$.}
\label{diam_stage2_fig}
\end{figure}

\begin{lemma}\label{diam_stage2_lem}
Let $n\in\N$ with \(n\leq2^{\sigma+1}k\) and \(2k<n\leq|X|\) for \(\sigma=(42d)^d\).
Then, on input $X$, Algorithm~\ref{diam_algo} computes a clustering \(\C_{2k}\) with
\[\diamcost(\C_{2k})<2^{3\sigma}\left(\diamcost(\C_n)+2\opt_k\right).\]
\end{lemma}
\begin{proof}
Let \(u:=\left\lceil\log_\frac{3}{4}\frac{2k}{n}\right\rceil\) and define \(m_i:=\left\lceil\left(\frac{3}{4}\right)^in\right\rceil\) for all $i=0,\ldots,u$.
Then, \(m_u\leq2k\) and \(m_i>2k\) for all $i=0,\ldots,u-1$.
Analogously to the proof of Lemma~\ref{diam_stage1_lem}, we get \(\left\lfloor\frac{3m_i}{4}\right\rfloor\leq m_{i+1}\) and using Lemma~\ref{diam_stage2_phaselem}, we deduce \(\diamcost(\C_{m_{i+1}})<2\cdot\left(\diamcost(\C_{m_i})+\opt_k\right))\) for all $i=0,\ldots,u-1$.
By repeatedly applying this inequality and using \(\diamcost(\C_{2k})\leq\diamcost(\C_{m_u})\), we get
\(\diamcost(\C_{2k})<2^u\cdot\left(\diamcost(\C_n)+2\opt_k\right)\).
Hence using \(u\leq\left\lceil\log_\frac{4}{3}2^\sigma\right\rceil<3\sigma\), the result follows.
\end{proof}

Proposition~\ref{diam_2k_prop} follows immediately by combining Lemma~\ref{diam_stage1_lem} and Lemma~\ref{diam_stage2_lem}.

\subsubsection{Analysis of the remaining merge steps}\label{diam_rem_sec}
We analyze the remaining merge steps analogously to the $k$-center problem.
Therefore, in this section we only discuss the differences, most of which are slightly modified bounds for the cost of merging two clusters (cf. Figure~\ref{diam_stage2_fig}).

The connectivity property from Section~\ref{connected_sec} remains the same.
However, for any two $(Z,r)$-connected clusters $A, B$, we use
\begin{equation}\label{diam_connect_eq}
\diam(A\cup B)\leq\diam(A)+\diam(B)+2r
\end{equation}
as a replacement for Inequality~(\ref{connect_eq}).
Furthermore, Lemma~\ref{connect_lem} also holds for the diameter $k$-clustering problem, i.e. with
\(\diamcost(\C_k)\leq\diamcost(\cP_\ell)\).

Using Inequality~(\ref{diam_connect_eq}) in the proof of Lemma~\ref{rad_rem_phaselem}, we get
\[\diam(A\cup B)\leq\diamcost(\cP_m)+\diamcost(\cP_n)+2\opt_k\]
as a replacement for Inequality~(\ref{rad_rem_phasebound1}) while Inequality~(\ref{rad_rem_phasebound2}) can be replaced by
\[\diam(A_1\cup A_2)\leq\diamcost(\cP_m)+2\cdot(\diamcost(\cP_n)+2\opt_k).\]
That is, for the diameter $k$-clustering problem the two upper bounds are different.
However, the second one is larger than the first one.
Using it in both cases, the inequality from Lemma~\ref{rad_rem_phaselem} changes slightly to
\[\diamcost(\cP_{\left\lfloor\frac{m+\ell}{2}\right\rfloor})\leq\diamcost(\cP_m)+2\cdot\left(\diamcost(\cP_n)+2\opt_k\right).\]
Together with
\(\diamcost(\cP_{\tilde{n}})\leq2\cdot\diamcost(\cP_n)+2\opt_k\)
as a replacement for Inequality~(\ref{rad_rem_lemeq}), the bound from Lemma~\ref{rad_rem_lem} becomes
\[\diamcost(\cP_{\ell+1})<2(\log_2(\ell)+2)\cdot\left(\diamcost(\cP_n)+2\opt_k\right).\]
Thus, using Proposition~\ref{diam_2k_prop}, the upper bound for the cost of the $\ell+1$-clustering of $Y$ from Lemma~\ref{rad_allstages} becomes
\[\diamcost(\cP_{\ell+1})<2(\log_2(\ell)+2)\cdot\left(2^{3\sigma}\left(28d+6\right)+2\right)\cdot\opt_k\] for \(\sigma=(42d)^d\).
Analogously to Section~\ref{rad_proof}, this proves Theorem~\ref{diam_result}.

\subsection{Analysis of the one-dimensional case}\label{upper_d1}
For $d=1$, we are able to show that Algorithm~\ref{diam_algo} computes an approximation to Problem~\ref{diam_prob} with an approximation factor of at most $3$.
We even know that for any input set $X\subset\R$ the approximation factor of the computed solution is strictly below $3$.
However, we do not show an approximation factor of $3-\epsilon$ for some $\epsilon>0$.
The proof of this upper bound is very technical, makes extensive use of the total order of the real numbers, and is certainly not generalizable to higher dimensions.
Therefore, we omit it.

\section{Lower bounds}
In this section, we present constructions of several input sets yielding lower bounds for the approximation factor of Algorithm~\ref{diam_algo}.
To this end, we look into possible runs of the algorithm.
Whenever Algorithm~\ref{diam_algo} is able to choose between several possible merge steps generating a cluster of equal minimum diameter, we simply assume that we can govern its choice.

In Section~\ref{lower_d1}, we show that for any input set $X\subset\R$ (i.e. $d=1$) Algorithm~\ref{diam_algo} has an approximation factor of at least $2.5$.
In Section~\ref{upper_d1}, we stated that in this case Algorithm~\ref{diam_algo} computes a solution to Problem~\ref{diam_prob} with approximation factor strictly below $3$.
Hence, for $d=1$, we obtain almost matching upper and lower bounds for the cost of the solution computed by Algorithm~\ref{diam_algo}.

Furthermore, in Section~\ref{lower_d2}, we show that the dimension $d$ has an impact on the approximation factor of Algorithm~\ref{diam_algo}.
This follows from a $2$-dimensional input set yielding a lower bound of $3$ for the metric based on the $\ell_\infty$-norm.
Note that this exceeds the upper bound from the one-dimensional case.

Moreover, in Section~\ref{lower_dvar}, we show that there exist input instances such that Algorithm~\ref{diam_algo} computes an approximation to Problem~\ref{diam_prob} with an approximation factor of $\Omega(\sqrt[p]{\log k})$ for metrics based on an $\ell_p$-norm $(1\leq p<\infty)$ and $\Omega(\log k)$ for the metric based on the $\ell_\infty$-norm.
In case of the $\ell_1$- and the $\ell_\infty$-norm, this matches the already known lower bound \cite{dasgupta} that has been shown using a rather artificial metric.
However, the bound in \cite{dasgupta} is derived from a $2$-dimensional input set, while in our instances the dimension depends on $k$.

Finally, we will see that the lower bound of $\Omega(\sqrt[p]{\log k})$ for any $\ell_p$-norm and $\Omega(\log k)$ for the $\ell_\infty$-norm can be adapted to the discrete $k$-center problem (see Section~\ref{drad_lower}).
In case of the $\ell_2$-norm, we thus obtain almost matching upper and lower bounds for the cost of the solution computed by Algorithm~\ref{drad_algo}.
Furthermore, we will be able to restrict the dependency on $d$ and $k$ of the approximation factor of Algorithm~\ref{drad_algo}.

\subsection{Any metric and $d=1$}\label{lower_d1}
We first show a lower bound for the approximation factor of Algorithm~\ref{diam_algo} using a sequence of input sets from $\R^d$ with $d=1$.
Since up to normalization there is only one metric for $d=1$, without loss of generality we assume the Euclidean metric.
\begin{proposition}
For all $\varepsilon>0$ and $k\geq4$ there exists an input set $X\subset\R$ such that Algorithm~\ref{diam_algo} computes a solution to Problem~\ref{diam_prob} with cost at least $\frac{5}{2}-\varepsilon$ times the cost of an optimal solution.
\end{proposition}

\begin{proof}
We show how to construct an input set for $k=4$.
The construction can easily be extended for $k>4$.
For any fixed $n\in\N$, we consider the following instance.
For $x\in\R$, we define a set $V(x)$ consisting of $2^n$ equidistant points:
\[V(x):=\{x+i\ |\ i\in\N\text{ and }0\leq i<2^n\}.\]
That is, neighboring points are at distance $1$ and \(\diam(V(x))=2^n-1\).
Furthermore, we define:
\begin{align*}
l(x)&:=x-2^{n-1},\\
r(x)&:=x+2^n-1+2^{n-1}=x+3\cdot2^{n-1}-1,\\
W(x)&:=V(x)\cup\{l(x),r(x)\}.
\end{align*}
It follows that \(\diam(W(x))=2^{n+1}-1\) as shown in Figure~\ref{lower_d1_sketchWx}.

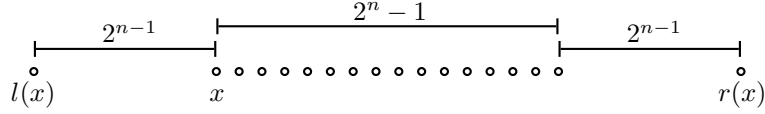
\begin{figure}[t]
\begin{center}
	\psset{xunit=0.3cm,yunit=0.3cm,runit=0.3cm}
	\begin{pspicture}(0,0)(31,4)

 		\pscircle(0,1){0.2}
 		\put(-1,-0.3){$l(x)$}
 		\multido{\ix=8+1}{16}{\pscircle(\ix,1){0.2}}
 		\put(7.7,-0.3){$x$}
 		\pscircle(31,1){0.2}
 		\put(30,-0.3){$r(x)$}

		\psline{|-|}(8,3)(23,3)
 		\put(14,3.3){$2^n-1$}
		\psline{|-|}(0,2)(8,2)
 		\put(3,2.3){$2^{n-1}$}
		\psline{|-|}(23,2)(31,2)
 		\put(26,2.3){$2^{n-1}$}

	\end{pspicture}
\end{center}
\caption{A sketch of the set $W(x)$.}
\label{lower_d1_sketchWx}
\end{figure}

We define the following input set $X$:
\[X:=\bigcup_{i=1}^4W(x_i)\]
where \(x_i:=i\cdot(7\cdot2^{n-1}-2)\) for $i=1,\ldots,4$.
Then, there is a gap of $3\cdot2^{n-1}-1$ between $W_n(x_i)$ and $W_n(x_{i+1})$, i.e.
\begin{equation}\label{gap}
\diam\left(\{r(x_i),l(x_{i+1})\}\right)=3\cdot2^{n-1}-1\quad\text{for }i=1,\ldots,3.
\end{equation}
The optimal $4$-clustering of $X$ is
\[\C_4^{opt}:=\{W(x_{1}),W(x_{2}),W(x_{3}),W(x_{4})\}\]
and \(\diamcost(\C_4^{opt})=2^{n+1}-1\).
However, the solution computed by Algorithm~\ref{diam_algo} may be worse.
At the beginning, the minimum distance between two points from $X$ is $1$.
The possible pairs of points with distance $1$ come from the sets $V(x_{i})$ for $i=1,\ldots,4$.
Since the distance between $V(x_{i})$ and $l(x_{i})$ or $r(x_{i})$ is $2^{n-1}$, we can assume that the algorithm merges all points of $V(x_{i})$ for $i=1,\ldots,4$ as shown in Figure~\ref{lower_d1_dendroW}.
It follows that Algorithm~\ref{diam_algo} computes the following $12$-clustering:
\begin{align*}
\C_{12}=\bigl\{&\{l(x_{1})\},V(x_{1}),\{r(x_{1})\},\\
&\{l(x_{2})\},V(x_{2}),\{r(x_{2})\},\\
&\{l(x_{3})\},V(x_{3}),\{r(x_{3})\},\\
&\{l(x_{4})\},V(x_{4}),\{r(x_{4})\}\bigr\}.
\end{align*}

\begin{figure}[t]
\begin{center}
	\psset{xunit=0.3cm,yunit=0.3cm,runit=0.3cm}
	\begin{pspicture}(0,0)(31,7)

 		\pscircle(0,4){0.2}
 		\put(-1,2.7){$l(x_{i})$}
 		\multido{\ix=8+1}{16}{\pscircle(\ix,4){0.2}}
 		\pscircle(31,4){0.2}
 		\put(30,2.7){$r(x_{i})$}

 		\multido{\nx=8+1}{16}{\psline(\nx,3.5)(\nx,3)}
 		\multido{\na=8+2,\nb=9+2}{8}{\psline(\na,3)(\nb,3)}
 		\multido{\nx=8.50+2.00}{8}{\psline(\nx,2.5)(\nx,2)}
 		\multido{\na=8.50+4.00,\nb=10.50+4.00}{4}{\psline(\na,2)(\nb,2)}
 		\multido{\nx=9.50+4.00}{4}{\psline(\nx,1.5)(\nx,1)}
 		\multido{\na=9.50+8.00,\nb=13.50+8.00}{2}{\psline(\na,1)(\nb,1)}
 		\multido{\nx=11.50+8.00}{2}{\psline(\nx,0.5)(\nx,0)}
		\psline(11.5,0)(19.5,0)

		\psline{|-|}(8,6)(23,6)
 		\put(14,6.3){$2^n-1$}
		\psline{|-|}(23,5)(31,5)
 		\put(26,5.3){$2^{n-1}$}
		\psline{|-|}(0,5)(15,5)
 		\put(3,5.3){$2^n-1$}

	\end{pspicture}
\end{center}
\caption{A part of the dendrogram for $W(x_{i})$.}
\label{lower_d1_dendroW}
\end{figure}
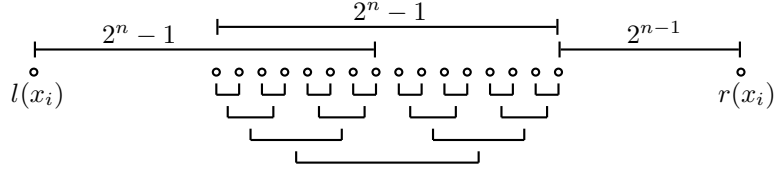

For $i=1,\ldots,4$, the diameters of \(\{l(x_{i})\}\cup V(x_{i})\) and \(V(x_{i})\cup\{r(x_{i})\}\) are equal to $3\cdot2^{n-1}-1$ and these are the best possible merge steps.
Therefore, by (\ref{gap}), we can assume that Algorithm~\ref{diam_algo} merges $r(x_{i})$ and $l(x_{i+1})$ for $i=1,\ldots,3$ first.
This results in the following $7$-clustering:
\begin{align*}
\C_7=\bigl\{&\{l(x_{1})\}\cup V(x_{1}),\\
&\{r(x_{1}),l(x_{2})\},V(x_{2}),\\
&\{r(x_{2}),l(x_{3})\},V(x_{3}),\\
&\{r(x_{3}),l(x_{4})\},V(x_{4})\cup\{r(x_{4})\}\bigr\}
\end{align*}
where $V(x_{2})$ and $V(x_{3})$ have a diameter of $2^n-1$ while the remaining clusters have a diameter of $3\cdot2^{n-1}-1$ (see Figure~\ref{lower_d1_dendroAll}).
Between two neighboring clusters of $\C_7$, there is a gap of $2^{n-1}$.

\begin{figure}[t]
\begin{center}
	\psset{xunit=0.12cm,yunit=0.12cm,runit=0.12cm}
	\begin{pspicture}(0,0)(93,16)
		\multido{\na=0+26,\nb=4+26,\nc=11+26,\nd=15+26}{4}{
			\pscircle(\na,10){0.2}
			\multido{\nx=\nb+1}{8}{\pscircle(\nx,10){0.2}}
			\pscircle(\nd,10){0.2}
			\psline(\nb,9)(\nb,8)
			\psline(\nc,9)(\nc,8)
			\psline(\nb,8)(\nc,8)
		}

		\psline(15,9)(15,8)
		\psline(26,9)(26,8)
		\psline(15,8)(26,8)

		\psline(41,9)(41,8)
		\psline(52,9)(52,8)
		\psline(41,8)(52,8)

		\psline(67,9)(67,8)
		\psline(78,9)(78,8)
		\psline(67,8)(78,8)

		\psline(00,9)(00,6)
		\psline(7.5,7)(7.5,6)
		\psline(00,6)(7.5,6)

		\psline(93,9)(93,6)
		\psline(85.5,7)(85.5,6)
		\psline(93,6)(85.5,6)

		\psline(20.5,7)(20.5,6)
		\psline(33.5,7)(33.5,6)
		\psline(20.5,6)(33.5,6)

		\psline(72.5,7)(72.5,6)
		\psline(59.5,7)(59.5,6)
		\psline(72.5,6)(59.5,6)

		\psline(3.75,5)(3.75,4)
		\psline(27,5)(27,4)
		\psline(3.75,4)(27,4)

		\psline(15.375,3)(15.375,2)
 		\put(14,-1){$C_1$}
		\psline(46.5,7)(46.5,2)
 		\put(45,-1){$C_2$}
		\psline(66,5)(66,2)
 		\put(64,-1){$C_3$}
		\psline(89.25,5)(89.25,2)
 		\put(88,-1){$C_4$}

		\psline{|-|}(0,12)(37,12)
 		\put(16,13){$\scriptstyle5\cdot2^n-3$}
		\psline{|-|}(41,12)(52,12)
 		\put(42,13){$\scriptstyle3\cdot2^{n-1}-1$}
		\psline{|-|}(56,12)(78,12)
 		\put(64,13){$\scriptstyle3\cdot2^n-2$}
		\psline{|-|}(82,12)(93,12)
 		\put(83,13){$\scriptstyle3\cdot2^{n-1}-1$}
	\end{pspicture}
\end{center}
\caption{A part of the dendrogram for $X$.}
\label{lower_d1_dendroAll}
\end{figure}
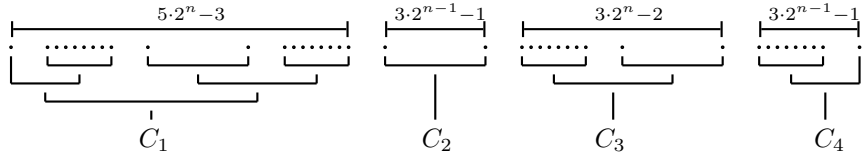

In the next step of Algorithm~\ref{diam_algo}, the best possible choice is to merge $\{r(x_{1}),l(x_{2})\}$ with $V(x_{2})$,  $\{r(x_{2}),l(x_{3})\}$ with $V(x_{2})$ or $V(x_{3})$, or $\{r(x_{3}),l(x_{4})\}$ with $V(x_{3})$.
We let it merge $\{r(x_{1}),l(x_{2})\}$ with $V(x_{2})$ and $\{r(x_{3}),l(x_{4})\}$ with $V(x_{3})$.
This results in a $5$-clustering where the clusters have alternating lengths of $3\cdot2^{n-1}-1$ and $3\cdot2^n-2$ with gaps of $2^{n-1}$ between them.
Then, in the step resulting in $\C_4$, Algorithm~\ref{diam_algo} has to create a cluster of diameter $5\cdot2^n-3$ as shown in Figure~\ref{lower_d1_dendroAll}.
Therefore, the computed solution has an approximation factor of
\[\frac{\diamcost(\C_4)}{\diamcost(\C_4^{opt})}=\frac{5\cdot2^n-3}{2^{n+1}-1}.\]
For $n$ going to infinity this approximation factor converges from below to $\frac{5}{2}$.
\end{proof}

\subsection{$\ell_\infty$-metric and $d=2$}\label{lower_d2}
In this section, we give a construction that needs only eight points from $\R^2$ and yields a lower bound of $3$ for the metric based on the $\ell_\infty$-norm.
Recall that in Section~\ref{upper_d1}, we showed that for $d=1$ the approximation factor of a computed solution is always strictly less than $3$.
Therefore, the lower bound of $3$ for $d=2$ implies that the dimension $d$ has an impact on the approximation factor of Algorithm~\ref{diam_algo}.

\begin{proposition}
For the metric based on the $\ell_\infty$-norm, there exists an input set $X\subset\R^2$ such that Algorithm~\ref{diam_algo} computes a solution to Problem~\ref{diam_prob} with three times the cost of an optimal solution.
\end{proposition}
\begin{proof}
We prove the proposition by constructing an example input set for $k=4$ (see Figure~\ref{lower_d2_fig}).
Consider the following eight points in $\R^2$:
\begin{alignat*}{2}
A&=(0,1),&\qquad E&=(-1,2),\\
B&=(1,0),&\qquad F&=(2,1),\\
C&=(0,-1),&\qquad G&=(1,-2),\\
D&=(-1,0),&\qquad H&=(-2,-1).
\end{alignat*}
The optimal $4$-clustering of these points is
\[\C_{4}^{opt}=\{\{A,E\},\{B,F\},\{C,G\},\{D,H\}\}\]
which has a maximum $\ell_\infty$-diameter of $1$.
However, it is also possible that Algorithm~\ref{diam_algo} starts by merging A with B and C with D.
Then, in the third step, the algorithm will merge $E$ or $F$ with $\{A,B\}$, $G$ or $H$ with $\{C,D\}$, or $\{A,B\}$ with $\{C,D\}$.
We assume the latter.
Thus, in the fourth merge step a cluster of $\ell_\infty$-diameter $3$ will be created.
\end{proof}

\begin{figure}[t]
\begin{minipage}[hbt]{.48\textwidth}
\begin{center}
	\psset{xunit=0.1cm,yunit=0.1cm,runit=0.1cm}
	\begin{pspicture}(-25,-25)(25,25)

 		\pscircle*(0,10){0.6}
 		\pscircle*(10,0){0.6}
 		\pscircle*(0,-10){0.6}
 		\pscircle*(-10,0){0.6}
 		\put(1,7){$A$}
 		\put(11,-3){$B$}
 		\put(1,-13){$C$}
 		\put(-9,-3){$D$}

 		\pscircle*(-10,20){0.6}
 		\pscircle*(20,10){0.6}
 		\pscircle*(10,-20){0.6}
 		\pscircle*(-20,-10){0.6}
 		\put(-9,17){$E$}
 		\put(21,7){$F$}
 		\put(11,-23){$G$}
 		\put(-19,-13){$H$}

		\psline{->}(-25,0)(25,0)
		\psline[linestyle=dotted]{-}(-25,-20)(25,-20)
		\psline[linestyle=dotted]{-}(-25,-10)(25,-10)
		\psline[linestyle=dotted]{-}(-25,10)(25,10)
		\psline[linestyle=dotted]{-}(-25,20)(25,20)
		\psline{->}(0,-25)(0,25)
		\psline[linestyle=dotted]{-}(-20,-25)(-20,25)
		\psline[linestyle=dotted]{-}(-10,-25)(-10,25)
		\psline[linestyle=dotted]{-}(10,-25)(10,25)
		\psline[linestyle=dotted]{-}(20,-25)(20,25)
		
	\end{pspicture}
\end{center}
\caption{Lower bound for the metric based on the $\ell_\infty$-norm.}
\label{lower_d2_fig}
\end{minipage}
\hfill
\begin{minipage}[hbt]{.48\textwidth}
\begin{center}
	\psset{xunit=0.08cm,yunit=0.08cm,runit=0.08cm}
	\begin{pspicture}(-35,-35)(35,35)

 		\pscircle*(-10,10){0.6}
 		\pscircle*(10,10){0.6}
 		\pscircle*(-10,-10){0.6}
 		\pscircle*(10,-10){0.6}
 		\put(-9,6){$A$}
 		\put(11,6){$B$}
 		\put(-9,-14){$C$}
 		\put(11,-14){$D$}

 		\pscircle*(-24,24){0.6}
 		\pscircle*(24,24){0.6}
 		\pscircle*(-24,-24){0.6}
 		\pscircle*(24,-24){0.6}
 		\put(-24,21){$E$}
 		\put(20,21){$F$}
 		\put(-28,-27){$G$}
 		\put(24,-27){$H$}

		\psline{->}(-35,0)(35,0)
		\psline[linestyle=dotted]{-}(-35,-30)(35,-30)
		\psline[linestyle=dotted]{-}(-35,-20)(35,-20)
		\psline[linestyle=dotted]{-}(-35,-10)(35,-10)
		\psline[linestyle=dotted]{-}(-35,10)(35,10)
		\psline[linestyle=dotted]{-}(-35,20)(35,20)
		\psline[linestyle=dotted]{-}(-35,30)(35,30)
		\psline{->}(0,-35)(0,35)
		\psline[linestyle=dotted]{-}(-30,-35)(-30,35)
		\psline[linestyle=dotted]{-}(-20,-35)(-20,35)
		\psline[linestyle=dotted]{-}(-10,-35)(-10,35)
		\psline[linestyle=dotted]{-}(10,-35)(10,35)
		\psline[linestyle=dotted]{-}(20,-35)(20,35)
		\psline[linestyle=dotted]{-}(30,-35)(30,35)
		
		\psarc[linestyle=dashed]{<->}(-10,10){20}{90}{180}
		\psarc[linestyle=dashed]{<->}(10,10){20}{0}{90}
		\psarc[linestyle=dashed]{<->}(-10,-10){20}{180}{270}
		\psarc[linestyle=dashed]{<->}(10,-10){20}{270}{0}

	\end{pspicture}
\end{center}
\caption{Lower bound for the metric based on the $\ell_2$-norm. The points $C,D,G,H$ have a $z$-coordinate of $0$, while the points $A,B,E,F$ have a $z$-coordinate of $2\sqrt{x}$.}
\label{lower_d3_fig}
\end{minipage}
\end{figure}

\subsection{Euclidean metric and $d=3$}\label{lower_d3}
For the Euclidean case, we are able to construct a $3$-dimensional instance that yields a lower bound of $2.56$.
This is below the upper bound of $3$ from the one-dimensional case.
Therefore, this instance does not show an impact of the dimension $d$ in the Euclidean case as in the previous section.
But this lower bound is still better than the lower bound of $2.5$ from the one-dimensional case.
This suggests that in higher dimensions it might be easier to construct good lower bounds.

\begin{proposition}
For the Euclidean metric there exists an input set $X\subset\R^3$ such that Algorithm~\ref{diam_algo} computes a solution to Problem~\ref{diam_prob} with cost $2.56$ times the cost of an optimal solution.
\end{proposition}
\begin{proof}
We prove the proposition by constructing an example input set for $k=4$ (see Figure~\ref{lower_d3_fig}).
For any fixed $x\in\R$ with $0<x<2$ consider the following eight points in $\R^2$:
\begin{alignat*}{8}
A&=(&-1&,&1&,&\ 2\sqrt{x}&),\qquad&E&=(&-(1+x)&,&1+\sqrt{4-x^2}\phantom{)}&,&\ 2\sqrt{x}&),\\
B&=(&1&,&1&,&2\sqrt{x}&),\qquad&F&=(&1+x\phantom{)}&,&1+\sqrt{4-x^2}\phantom{)}&,&2\sqrt{x}&),\\
C&=(&-1&,&\ -1&,&0&),\qquad&G&=(&-(1+x)&,&\ -(1+\sqrt{4-x^2})&,&0&),\\
D&=(&1&,&-1&,&0&),\qquad&H&=(&1+x\phantom{)}&,&-(1+\sqrt{4-x^2})&,&0&).
\end{alignat*}
The optimal $4$-clustering of these points is
\[\C_{4}^{opt}=\{\{A,E\},\{B,F\},\{C,G\},\{D,H\}\},\]
which has a maximum $\ell_2$-diameter of $2$.
However, since $\|A-B\|=\|C-D\|=2$ it is possible that Algorithm~\ref{diam_algo} starts by merging $A$ with $B$ and $C$ with $D$.
Then, the cheapest merge adds one of the points $E,F$ to the cluster $\{A,B\}$ or it adds one of the points $G,H$ to the cluster $\{C,D\}$ or it merges $\{A,B\}$ with $\{C,D\}$.
We assume the latter.
The resulting cluster $\{A,B,C,D\}$ has a diameter of $2\sqrt{2+x}$.
Then, in the fourth merge step, the algorithm will either merge one of the pairs $E,F$ and $G,H$ or one of the pairs $E,G$ and $F,H$.
The choice depends on the parameter $x$.
Note that Algorithm~\ref{diam_algo} will not merge the cluster $\{A,B,C,D\}$ with one of the remaining four points, since this is always more expensive.
The diameter of the created cluster is maximized for $x\approx1.56$.
If we fix $x=1.56$, the algorithm merges $E$ with $F$ or $G$ with $H$.
This results in a $4$-clustering of cost $5.12$, while the optimal solution has cost $2$.
\end{proof}

\subsection{$\ell_p$-metric ($1 \leq p \leq \infty$) in variable dimension}\label{lower_dvar}
In the following, we consider the diameter $k$-clustering problem with respect to the metric based on the $\ell_1$-norm.
We show that there exists an input instance in dimension $O(k)$ such that Algorithm~\ref{diam_algo} computes a solution with an approximation factor of $\Omega(\log k)$.

\begin{proposition}\label{lower_dvar_diam_prop}
For the metric based on the $\ell_1$-norm, there exists an input set $X \subset \R^d$ with $d=k+\log_2k$ such that Algorithm~\ref{diam_algo} computes a solution to Problem~\ref{diam_prob} with $\frac1{2}\log_2k$ times the cost of an optimal solution.
\end{proposition}
\begin{proof}
For simplicity's sake, assume $k$ to be a power of 2.
In the sequel, we consider the $(k+\log_2k)$-dimensional set $X$ of $|X|=k^2$ points defined by
\[X = \left\{ \left[ \begin{matrix}e_i\\b\end{matrix} \right] \,\vline~ \forall 1 \leq i \leq k ~\mbox{and}~ b \in \{0,1\}^{\log_2k} \right\} .\]
Here, $e_i \in \R^k$ denotes the $i$-th canonical unit vector.
Consider the following $k$-clustering 
\[\C_k^* = \left\{ C_b \where b \in \{0,1\}^{\log_2k} \right\},\]
where for each $b \in \{0,1\}^{\log_2k}$ cluster $C_b$ is given by
\[C_b = \left\{ \left[ \begin{matrix}e_i\\b\end{matrix} \right] \where \forall 1 \leq i \leq k \right\}.\]
The largest diameter of $\C_k^*$ is \(\diamcost(\C_k^*)=2\).
Hence for $\opt_k$, the diameter of an optimal solution, it holds
\begin{equation}\label{lower_dvar_diam_eq}
\opt_k\leq2.
\end{equation}
However, we find that
\[\diam\left( \left\{ \left[ \begin{matrix}e_i\\b_1\end{matrix} \right], \left[ \begin{matrix}e_j\\b_2\end{matrix} \right]  \right\} \right) = \begin{cases}h(b_1,b_2)&\text{if $i=j$}\\2 + h(b_1,b_2)&\text{if $i\not=j$}\end{cases}\]
where $h(b_1,b_2)$ denotes the Hamming distance between the strings $b_1,b_2 \in \{0,1\}^{\log_2k}$.
Hence, we may assume that Algorithm~\ref{diam_algo} starts by merging points $[e_i,0,b']^\top$ and $[e_i,1,b']^\top$ for all $1 \leq i \leq k$ and all $b' \in \{0,1\}^{\log_2(k)-1}$, thereby forming $\frac1{2}k^2$ clusters of diameter 1.

Next, we show inductively that Algorithm~\ref{diam_algo} keeps merging pairs of clusters that agree on the first $k$ coordinates until the algorithm halts.
To this end, assume that there is some number $1 \leq t \leq \log_2k$ such that the clustering computed so far consists solely of the clusters
\[C_{i,b'}^{(t)} = \left\{ \left[ \begin{matrix}e_i\\b\\b'\end{matrix} \right] \where b \in\{0,1\}^t \right\}\]
for all $1 \leq i \leq k$ and all $b'\in\{0,1\}^{\log_2(k)-t}$.
Also note that this is the case with $t=1$ after the first $\frac1{2}k^2$ merges.
In such a case, we have
\[\diam\left( C_{i,b_1}^{(t)} \cup C_{j,b_2}^{(t)}\right) = \begin{cases}t + h(b_1,b_2)&\text{if $i=j$}\\2 + t + h(b_1,b_2)&\text{if $i\not=j$}\end{cases}.\]
Hence, as above, we may assume that in the next $\frac1{2^{t+1}}k^2$ steps Algorithm~\ref{diam_algo} merges the clusters $C_{i,0b'}^{(t)}$ and $C_{i,1b'}^{(t)}$ for all $1 \leq i \leq k$ and all $b' \in \{0,1\}^{\log_2(k)-(t+1)}$.
The resulting clusters are of diameter $t+1$.
Also, we have $C_{i,b'}^{(t+1)} = C_{i,0b'}^{(t)} \cup C_{i,1b'}^{(t)}$.

Algorithm~\ref{diam_algo} keeps merging clusters in this way until after $t=\log_2k$ rounds we end up with the $k$-clustering $\C_k = \{C_i \where 1 \leq i \leq k\}$ where
\[C_i = \left\{ \left[ \begin{matrix}e_i\\b\end{matrix} \right] \where b \in\{0,1\}^{\log_2k} \right\} ~.\]
These clusters $C_i$ are of diameter $\log_2k$.
Comparing to (\ref{lower_dvar_diam_eq}), we deduce that Algorithm~\ref{diam_algo} computes a solution to Problem~\ref{diam_prob} with at least $\frac{1}{2}\log_2k$ times the cost of an optimal solution.
\end{proof}

Considering the diameter $k$-clustering problem with respect to an arbitrary $\ell_p$-metric (with $1 \leq p < \infty$), note that the behavior of Algorithm~\ref{diam_algo} does not change if we consider the $p$-th power of the $\ell_p$-distance instead of the $\ell_p$-distance.
Also note that for all $x,y \in \{0,1\}^d$ we have $\|x-y\|_p^p = \|x-y\|_1$.
Since instance $X$ from Proposition~\ref{lower_dvar_diam_prop} is a subset of $\{0,1\}^d$, we immediately obtain the following corollary.

\begin{corollary}\label{lower_dvar_diam_lpcor}
For the metric based on any $\ell_p$-norm with $1 \leq p < \infty$, there exists an input set $X \subset \R^d$ with $d=k+\log_2k$ such that Algorithm~\ref{diam_algo} computes a solution to Problem~\ref{diam_prob} with $\sqrt[p]{\frac1{2}\log_2k}$ times the cost of an optimal solution.
\end{corollary}

Additionally, considering the diameter $k$-clustering problem with respect to the $\ell_\infty$-metric, it is known that every $n$-point subset of an arbitrary metric space can be embedded isometrically into $(\R^n,\ell_\infty)$  \cite{Freche10}.
Hence, the instance from Proposition~\ref{lower_dvar_diam_prop} of size $n=k^2$ yields an instance in $\R^{k^2}$ satisfying the same approximation bound with respect to the $\ell_\infty$-distance.
We obtain the following corollary.

\begin{corollary}\label{lower_dvar_diam_lmaxcor}
For the metric based on the $\ell_\infty$-norm, there exists an input set $X \subset \R^d$ with $d=k^2$ such that Algorithm~\ref{diam_algo} computes a solution to Problem~\ref{diam_prob} with $\frac1{2}\log_2k$ times the cost of an optimal solution.
\end{corollary}

\subsubsection{The discrete $k$-center problem}\label{drad_lower}
The input instance $X$ from Proposition~\ref{lower_dvar_diam_prop} also proves lower bounds on the approximation factor of the agglomerative solution to the discrete $k$-center problem.
To this end, just note that for the instance $X$ in every step of the algorithm the minimal discrete radius of a cluster equals the diameter of the cluster.
We immediately obtain the following corollaries.

\begin{corollary}\label{lower_dvar_drad_lpcor}
For the metric based on any $\ell_p$-norm with $1 \leq p < \infty$, there exists an input set $X \subset \R^d$ with $d=k+\log_2k$ such that Algorithm~\ref{drad_algo} computes a solution to Problem~\ref{drad_prob} with $\sqrt[p]{\frac1{2}\log_2k}$ times the cost of an optimal solution.
\end{corollary}
\begin{corollary}\label{lower_dvar_drad_lmaxcor}
For the metric based on the $\ell_\infty$-norm, there exists an input set $X \subset \R^d$ with $d=k^2$ such that Algorithm~\ref{drad_algo} computes a solution to Problem~\ref{drad_prob} with $\frac1{2}\log_2k$ times the cost of an optimal solution.
\end{corollary}

Moreover, in case of the $\ell_2$-norm, we obtain the following corollary.
\begin{corollary}\label{lower_dvar_drad_l2cor}
For the metric based on the $\ell_2$-norm, there exists an input set $X \subset \R^d$ with $d=O(\log^3k)$ such that Algorithm~\ref{drad_algo} computes a solution to Problem~\ref{drad_prob} with $\Omega(\sqrt{\log k})$ times the cost of an optimal solution.
\end{corollary}

Corollary~\ref{lower_dvar_drad_l2cor} follows by embedding the instance from Corollary~\ref{lower_dvar_drad_lpcor} into the $O(\log^3k)$-dimensional Euclidean space without altering the behavior of the agglomerative algorithm or the lower bound of $\Omega(\sqrt{\log k})$ (Johnson-Lindenstrauss embedding \cite{JohLin84}).
For this embedded instance, the bound given in Section~\ref{drad_sec} states an upper bound of $20d+2\log(k)+2=O(\log^3k)$ times the cost of an optimal solution.
Hence, in case of the discrete $k$-center clustering using the $\ell_2$-metric, the upper bound from our analysis almost matches the lower bound.

Furthermore, this implies that the approximation factor of Algorithm~\ref{drad_algo} cannot be simultaneously independent of $d$ and $\log k$.
More precisely, the approximation factor cannot be sublinear in $\sqrt[6]{d}$ and in $\sqrt{\log k}$.

\section{Open problems}
The main open problems our work raises are:
\begin{itemize}
\item Can the doubly exponential dependence on $d$ in Theorem~\ref{diam_result} be improved?
\item Are the different dependencies on $d$ in the approximation factors for the discrete $k$-center problem, the $k$-center problem, and the diameter $k$-clustering problem due to the limitations of our analysis or are they inherent to these problems?
\item Can our results be extended to more general distance measures?
\item Can the lower bounds for $\ell_p$-metrics with $1<p<\infty$ be improved to $\Omega(\log k)$, matching the lower bound from \cite{dasgupta} for all $\ell_p$-norms?
\end{itemize}

\bibliographystyle{plain}
\bibliography{references}

\end{document}